\documentclass{marticle}

\title{Expanderizing Higher Order Random Walks}
\author{Vedat Levi Alev\thanks{Hebrew University of Jerusalem -- 
\href{mailto:vedatle.alev@mail.huji.ac.il}{vedatle.alev@mail.huji.ac.il}} \and Shravas Rao\thanks{Portland State University -- 
\href{mailto:shravas@pdx.edu}{shravas@pdx.edu}}}
\date{\today}
\begin{document}
\begin{titlepage}
    \def\thepage{}
    \thispagestyle{empty}

    \maketitle
    \begin{abstract}
	   We study a variant of the down-up (also known as the Glauber dynamics) and up-down walks over an $n$-partite simplicial complex, which we call  \emph{expanderized higher order random walks} -- where the sequence of updated coordinates correspond to the sequence of vertices visited by a random walk over an auxiliary expander graph $H$. When $H$ is the clique with self loops on $[n]$, this random walk reduces to the usual down-up walk and when $H$ is the directed cycle on $[n]$, this random walk reduces to the well-known systematic scan Glauber dynamics. We show that whenever the usual higher order random walks satisfy a log-Sobolev inequality or a Poincar\'e inequality, the expanderized walks satisfy the same inequalities with a loss of quality related to the two-sided expansion of the auxillary graph $H$. Our construction can be thought as a higher order random walk generalization of the derandomized squaring algorithm of Rozenman and Vadhan (RANDOM 2005).
    
    We study the mixing times of our expanderized walks in two example cases: We show that when initiated with an expander graph our expanderized random walks have mixing time (i) $O(n \log n)$ for sampling a uniformly random list colorings of a graph $G$ of maximum degree $\Delta = O(1)$ where each vertex has at least $(11/6 - \ee) \Delta$ and at most $O(\Delta)$ colors, (ii) $O_{\etch}\parens*{ \frac{n \log n}{(1 - \norm*{\Jay}_\opp)^2}}$ for sampling the Ising model with a PSD interaction matrix $\Jay \in \RR^{n \times n}$ satisfying $\norm{\Jay}_\opp \le 1$ and the external field $\etch \in \RR^n$-- here the $O(\bullet)$ notation hides a constant that depends linearly on the largest entry of $\etch$. As expander graphs can be very sparse, this decreases the amount of randomness required to simulate the down-up walks by a logarithmic factor.
    
        We also prove some simple results which enable us to argue about log-Sobolev constants of higher order random walks and provide a simple and self-contained analysis of local-to-global $\Phi$-entropy contraction  in simplicial complexes -- giving simpler proofs for many pre-existing results. 
    \end{abstract}
\end{titlepage}
\newpage

\section{Introduction}
\subsection{Higher Order Random Walks, Systematic Scan, and Expanderized Walks}
\footnote{All concepts and random walks we define in this introductory section of the paper, will be formally defined in \cref{sec:prelim} and \cref{sec:expander}}Let $U_1, \cdots, U_n$ be a collection of finite sets. The \emph{down-up walk $\Duw$} on $\Omega \subset U_1 \times \cdots U_n$ with respect to a given distribution $\pi: \Omega \to \RRp$, also known as the \emph{Glauber dynamics} on $\Omega$ according to $\pi$, is the following simple process: Starting from an arbitrary tuple $\omega^{(0)}$, we obtain the $(t+1)$-st tuple $\omega^{(t+1)}$ visited by this random walk from the $t$-th tuple $\omega^{(t)}$ as follows,\\
\begin{boxedc}{0.85\textwidth}{ Update Rule for the Down-Up Walk, $\Duw$}
\begin{enumerate}
    \item sample a uniformly random coordinate $i \sim \uni_{[n]}$,
    \item sample a random tuple $\omega^{(t+1)} \sim \pi$ conditional on $\omega^{(t+1)}_j = \omega^{(t)}_j$ for all $j \in [n] \setminus \set i$.
\end{enumerate}
\end{boxedc}

The following variant of the down-up walk, called the \emph{systematic scan $\Psc$} on $\Omega$ according to $\pi$, is a variant of the down-up walk $\Duw$ which uses less randomness and is easier to implement in practice: starting from an arbitrary tuple $\omega^{(0)}$, we obtain the $(t+1)$-st tuple $\omega^{(t+1)}$ visited by this random walk from the $t$-th tuple $\omega^{(t)}$ as follows,
\begin{boxedc}{0.85\textwidth}{Update Rule for the  Systematic Scan, $\Psc$}
\begin{enumerate}
    \item set $i = t + 1 \pmod n$,
    \item sample a random tuple $\omega^{(t+1)} \sim \pi$ conditional on $\omega^{(t+1)}_j = \omega^{(t)}_j$ for all $j \in [n] \setminus \set i$.
\end{enumerate}
\end{boxedc}
In both cases, the coordinate $i$ that is sampled on the first step of the update can be thought as a vertex visited by the simple random walk on a graph. For the down-up walk, this is a random walk on the clique with self-loops, whereas for the systematic scan this is a (deterministic) walk on the directed cycle. 

The main object of study in this paper will be the so-called \emph{expanderized down-up walk} $\Paqx$ on $\Omega$ with respect to the distribution $\pi: \Omega \to \RRpp$ and the $k$-regular graph $H = ([n], E)$ for some constant $k$. Starting this random-walk from an arbitrary coordinate $i^{(0)} \in [n]$ and an arbitrary tuple $\omega^{(0)}$, we obtain the $(t+1)$-st coordinate $i^{(t+1)}$ and tuple $\omega^{(t+1)}$ according to the following update rule,
\begin{boxedc}{0.9 \textwidth}{Update Rule for the Expanderized Down-Up Walk $\Paqx$}
    \begin{enumerate}
        \item sample a random neighbor $s$ of $i^{(t)}$ in $H$,
        \item sample a random tuple $\omega^{(t+1)} \sim \pi$ conditional on $\omega^{(t+1)}_j = \omega^{(t)}_j$ for all $j \in [n] \setminus \set s$,
        \item set $i^{(t+1)}$ to be a random neighbor of $s$ in $H$.
    \end{enumerate}
\end{boxedc}
We notice that according to the above update rule when $i^{(0)}$ is sampled uniformly at random and $H$ equals the clique with self-loops on $[n]$ the evolution of $\omega^{(t)}$ is as dictated by the down-up walk $\Duw$. Similarly, when $i^{(0)} = 1$ and $H$ is the directed cycle, the evolution of $\omega^{(t)}$ is as dictated by the systematic scan $\Psc$. 

The main contribution of this paper is an analysis of the expanderized down-up walk assuming, (i) the graph $H$ is a spectral expander\footnote{i.e.~all non-trivial eigenvalues of $H$ are bounded away from 1} and (ii) the down-up walk $\Duw$ satisfies some kind of isoperimetric inequality, e.g.~a log-Sobolev inequality or a Poincar\'e inequality. Indeed our methods allow us to extend our results to all down-up and up-down walks.
\subsection{Motivation: Systematic Scan, Expander Graphs, and Derandomized Squaring}
The systematic scan $\Psc$ is a random walk of great practical and theoretical interest. Yet, rapid mixing results for this walk are only known under restricted circumstances \cite{DiaconisR00,Hayes06, DyerGJ06,DyerGJ08, RobertsR15, FengGWWY23} and it is very hard to directly relate the rapid mixing of $\Duw$ to that of $\Psc$. A particularly useful framework for establishing rapid mixing for the down-up walk is the method of high-dimensional expansion, in particular the frameworks of spectral independence and entropic independence \cite{AlevL20, AnariLO20, ChenLV20, ChenLV21, ChenE22, AnariJKP22,AnariJKP22a, AnariJKPV23} which led to many breakthrough results in the field of sampling algorithms. 

In \cite{AlevP23}, an attempt was made to study the mixing of the systematic scan\footnote{More formally, $n$ successive steps of the systematic scan, which the authors call the \emph{sequential sweep} $\Psq$.} using techniques of high-dimensional expansion -- while their techniques allowed them to establish rapid mixing results for constant dimensional partite simplicial complexes, their result is too restrictive to take advantage of mixing results obtained through spectral independence or entropic independence.  As a step towards directly being able to take advantage of the mixing results for $\Duw$, which could potentially be obtained through the high-dimensional expansion framework, we introduce our expanderized down-up walks $\Paqx$. As expander graphs have proven themselves very successful at approximating dense objects, we hope -- and indeed also prove -- that transfering mixing time bounds from the usual down-up walks to our expanderized walks to be an easier task than establishing mixing times for $\Psc$.  As expander graphs can be very sparse, our expanderized walks can be thought as replacing the sparse object used in the definition of the systematic scan $\Psc$, i.e.~the directed cycle, with another sparse yet highly connected object -- an expander graph with constant degree. 

In spirit, the expanderized walks can be thought as a higher order random walk analogue of the derandomized squaring algorithm introduced in \cite{RozenmanV05}. This algorithm was introduced to simplify the seminal result of \cite{Reingold08} concerning the existence of a logspace algorithm for deciding undirected connectivity. The derandomized squaring operation uses an auxiliary $k$-regular expander graph $H$ on the vertex set $[d]$ to approximate the square of a graph $d$-regular graph $G$ on $[n]$. Whereas the actual square $G^2$ is a $d^2$-regular graph, by picking $k = O(1)$ one can ensure that the \emph{derandomized square} is $O(d)$-regular, i.e.~a much sparser object. This result rests on the observation that the actual square $G^2$ is obtained from the graph $G$ by attaching a clique to every vertex -- replacing this clique with an expander graph suffices to ensure that the resulting \emph{derandomized square} is closed to the actual square. Fortunately, the same intuition also leads to proofs showing that the expanderized walks approximate the standard walks well.
\subsection{Our Results}
\subsubsection*{Expanderized Walks}
Our main contribution in this paper is the study and analysis of expanderized higher order random walks. Since throughout the paper we use the language of simplicial complexes, we recall some basics: A simplicial complex $X$ is a downward closed collection of subsets of some finite set $U$. We write $X^{(j)}$ to denote the subsets in $X$ of size $j$ -- we call the cardinality of the largest element in $X$ the rank of $X$ and elements of $X$ its faces. We call $(X, \pi)$ a weighted simplicial complex of rank $n$, where $\pi: X^{(n)} \to \RRpp$ is a probability distribution on $X^{(n)}$. Throughout, we will assume that $X$ is obtained by taking the downward closure of some collection of interest, i.e.~we have $X^{(n)} = \Omega$ where $\Omega$ is a collection of $n$-elemented sets and $X^{(j)}$ is the collection of $j$-elemented subsets which are contained in $\omega \in \Omega$ for all $j = 0, \ldots, n-1$. We will say $X$ is an $n$-partite simplicial complex, if it is of rank $n$ and there exists some finite sets  $U_1,\ldots, U_n$ such that $X^{(n)}$ can be identified with a subset of $U_1 \times \cdots \times U_n$. For example by, identifying the tuple $(u_1,\ldots,u_n)$ with the set $\set*{(1, u_1), \ldots, (n, u_n)}$. The sets $U_1,\ldots, U_n$ are called the sides of the simplicial complex. For $\omega \in X^{(n)}$ and $S \subset [n]$, we will write $\omega_S$ for the restriction of $\omega$ to the coordinates in $S$, i.e.~if $\omega = (u_1, \ldots, u_n)$ we identify $\omega_S$ with the set $\set*{ (s, u_s) \mid s \in S}$. For now, we will restrict our attention to partite simplicial complexes. In particular, every element in $\widehat\omega \in X^{(\ell)}$ in a partite simplicial complex can be obtained from a face $\omega \in X^{(n)}$ and some $S \in \binom{[n]}{\ell}$, by setting $\widehat\omega = \omega_{S}$. Whereas the choice of $\omega$ is not unique, the choice of $S$ is. For $\widehat\omega \in X^{(< n)}$ we denote this unique choice by $\typ(\widehat\omega)$.

We now recall the following random walks between $X^{(n)}$ and $X^{(\ell)}$. The first random walk is the so-called down-up walk $\Duw_{n \lra \ell}$, which is also known as the Glauber dynamics with block size $(n-\ell)$. Starting from an arbitrary face $\omega^{(0)} \in X^{(n)}$, we obtain the $t$-th face $\omega^{(t)} \in X^{(n)}$ that we visit, from the $(t-1)$-st face $\omega^{(t-1)}$ according to the following update rule, 
\begin{boxedc}{0.85\textwidth}{ Update Rule for the Down-Up Walk According to $\pi$, $\Duw_{n \lra \ell}$}
\begin{enumerate}
    \item sample a uniformly random set of coordinates $S \sim \uni_{\binom{[n]}{\ell}}$,
    \item sample a random tuple $\omega^{(t+1)} \sim \pi$ conditional on $\omega^{(t+1)}_S = \omega^{(t)}_S$.
\end{enumerate}
\end{boxedc}
We notice that for $\ell = n - 1$, this is the same walk we have defined in the preceding section.

We also recall the up-down walk $\Udw_{\ell \lra n}$ between  $X^{(\ell)}$ and $X^{(n)}$ (according to $\pi$) -- where $(X,\pi)$ is a partite simplicial complex. This is the random walk that starts from an arbitrary face $\widehat\omega^{(0)} \in X^{(\ell)}$, and determines the $t$-th face $\widehat\omega^{(t)}$ visited by this random walk using the $(t-1)$-st face $\widehat\omega^{(t-1)}$ according to the following update rule,
\begin{boxedc}{0.85\textwidth}{ Update Rule for the Up-Down Walk According to $\pi$, $\Duw_{\ell \lra n}$}
\begin{enumerate}
    \item sample a random face $\omega \sim \pi$ conditional on containing $\widehat\omega^{(t-1)}$,
    \item sample a uniformly random set $S \sim \uni_{\binom{[n]}{\ell}}$, 
    \item output $\widehat\omega^{(t)} = \omega_{S}$.
\end{enumerate}
\end{boxedc}
 Given a degree regular graph $H$ on the vertex set $\binom{[n]}{\ell}$, we define the expanderized up-down walk between $X^{(\ell)}$ and $X^{(n)}$ (via $H$ according to $\pi$), as a random walk which starts from an arbitrary face $\widehat\omega^{(0)} \in X^{(\ell)}$ and picks the $t$-th face visited $\widehat\omega^{(t)}$ by this random walk using the $(t-1)$-st face $\widehat\omega^{(t-1)}$ according to the following update rule, 
\begin{boxedc}{0.85\textwidth}{ Update Rule for the Expanderized Up-Down Walk via $H$ according to $\pi$, $\Papx_{\ell \lra n}$}
\begin{enumerate}
    \item sample a random face $\omega \sim \pi$ conditional on containing $\widehat\omega^{(t-1)}$,
    \item sample a random neighbor $S$ of $\typ(\widehat\omega^{(t-1)})$ in $H$,
    \item output $\widehat\omega^{(t)} = \omega_S$.
\end{enumerate}
\end{boxedc}
Finally, we define the expanderized down-up walk $\Paqx_{n \lra \ell}$ between $X^{(n)}$ and $X^{(\ell)}$ according to $\pi$ as the following random walk on $X^{(n)} \times \binom{[n]}{\ell}$, starting from an arbitrary pair $(\omega^{(0)}, S^{(0)}) \in X^{(n)} \times \binom{[n]}{\ell}$, this walk picks the $t$-th pair visited by this random walk as follows,
\begin{boxedc}{0.9 \textwidth}{Update Rule for the Expanderized Down-Up Walk $\Paqx_{n \lra \ell
}$}
    \begin{enumerate}
        \item sample a random neighbor $T$ of $S^{(t-1)}$ in $H$,
        \item sample a random tuple $\omega^{(t)} \sim \pi$ conditional on $\omega^{(t)}_T = \omega^{(t-1)}_T$,
        \item set $S^{(t)}$ to be a random neighbor of $T$ in $H$.
    \end{enumerate}
\end{boxedc}

\begin{remark}
    Once the expanderized down-up walk reaches stationarity, the face component will be distributed according to $\pi$. Similarly, once the expanderized up-down walk reaches stationarity, it will be distributed according to the correct marginal distribution of $\pi$. Both random walks are reversible. 
\end{remark}

We also recall that the \ref{eq:gap-def} $\Gap(\Pii)$ of a reversible random walk matrix $\Pii \in \RR^{\Omega \times \Omega}$ is defined to be,
\[ \Gap(\Pii) = 1 - \lambda(\Pii),\]
where $\lambda(\Pii)$ denotes the \ref{eq:two-sided} expansion of $\Pii$, i.e.~$\lambda(\Pii) = \max\set*{ \lambda_2(\Pii), | \lambda_{\min}(\Pii)|}$.\footnote{In the main body of the paper we will adopt the convention $\Gap(\Pii) = 1 - \lambda_2(\Pii)$, but to keep our exposition simple we will work with the \ref{eq:two-sided} parameter $\lambda(\Pii)$ throughout the introduction.}

It is well known that a bound on the \ref{eq:gap-def} translates into a bound on the mixing time of the random walk $\Pii$. The following bound is well known, cf.~\cite[Proposition 1.12]{MontenegroT05}
\begin{theorem}
    Let $\Pii \in \RR^{\Omega \times \Omega}$ be a reversible random walk matrix with stationary distribution $\pi: \Omega \to \RRpp$, i.e.~$\pi\Pii = \pi$. We have,
    \[ \Tmix(\Pii, \ee) \le \frac{1}{1 - \lambda(\Pii)} \cdot \log\frac{1}{\ee \cdot \sqrt{\min_{\omega \in \Omega} \pi(\omega)}}.\] 
\end{theorem}

We now describe our results concerning expanderized random walks, which we will be proven in~\cref{sec:expander}. 
\begin{theorem}[Simplified Version of \cref{cor:gaplift}]\label{thm:gshow}
    Let $(X, \pi)$ be an $n$-partite simplicial complex and $H$ a degree regular graph on the vertex set $\binom{[n]}{\ell}$ where $0 \le \ell \le n$. Writing $\Duw_{n \lra\ell}$ and $\Udw_{\ell \lra n}$ for the down-up walk and up-down walks between $X^{(\ell)}$ and $X^{(n)}$  according to $\pi$, and $\Papx_{\ell \lra n}$ for the expanderized up-down walk via $H' = H^2$ and $\Paqx_{n \lra \ell}$ for the expanderized down-up walk via $H$, we have
    \[ \Gap\parens*{ \Paqx_{n \lra \ell} } \ge \Gap\parens*{ \Duw_{n \lra \ell} } \cdot \Gap(H^2)~~\textrm{and}~~\Gap\parens*{\Papx_{\ell \lra n}} \ge \Gap\parens*{ \Udw_{\ell \lra n} } \cdot \Gap(H).\]
\end{theorem}

\begin{remark}
    We note that the difference in the dependence on $\Gap(H)$ between the bounds of $ \Gap\parens*{ \Paqx_{n \lra \ell} }$ and $\Gap\parens*{\Papx_{\ell \lra n}}$ is an artifact of the update rule of the two walks, which we have chosen to ensure that the down- and up-movements of the expanderized walks to be adjoints of each other. In particular, in the expanderized down-up walk a sequence of two random vertices in $H$ is sampled, while in the expanderized up-down only one vertex is sampled.  
\end{remark}
Similarly, we recall that the \ref{eq:ent-def} functional $\Ent_\pi(\eff)$ is defined as,\[ \Ent_\pi(\eff) = \Exp_{\omega \sim \pi} \eff(\omega) \log \eff(\omega) - \parens*{ \Exp_\pi \eff}\log\parens*{ \Exp_\pi \eff}.\]
We recall that the entropy contraction constant $\EC(\Pii)$ and the log-Sobolev constant $\LS(\Pii)$ of a reversible  random walk $\Pii \in \RR^{\Omega \times \Omega}$ are the largest constants $C_{ec} \ge 0$ and $C_{ls} \ge 0$ respectively, satisfying the following inequalities for all $\eff \in \RRp^{\Omega}$,
\begin{align*}
    \Ent_\pi(\Pii\eff) &~\le~(1 - C_{ec}) \cdot \Ent_\pi(\eff),\\
    C_{ls} \cdot \Ent_\pi(\eff^2) &~\le~\mathcal E_{\Pii}(\eff, \eff)
\end{align*}
where we have written $\mathcal E_\Pii(\bullet, \bullet)$ for the Dirichlet form of the random walk $\Pii$. We recall that the entropy contraction constant $\EC(\Pii)$ allows establishing a sharper bound on the mixing time,
\begin{theorem}[\cite{BlancaCPSV21}]
    There exists a universal constant $C$ such that, for any reversible random walk $\Pii \in \RR^{\Omega \times \Omega}$ with stationary distribution $\pi: \Omega \to \RRpp$, i.e.~$\pi \Pii = \pi$. We have
    \[ \Tmix(\Pii, \ee) \le \frac{C}{\EC(\Pii)} \cdot  \parens*{{\log\log\frac{1}{\min_{\omega \in \Omega} \pi(\omega)} + \log \ee^{-1}}},\]
    where the constant $C$ does not depend on the pair $(\Pii, \pi)$ and $\EC(\Pii)$ denotes the entropy contraction constant of the random walk $\Pii$.
\end{theorem}
We can prove the following bounds for the entropy contraction of expanderized walks,
\begin{theorem}[Simplified Version of \cref{cor:entropicstuff}]\label{thm:eshow}
    Let $(X, \pi)$ be a weighted $n$-partite simplicial complex and $H$ a degree-regular graph on $\binom{[n]}{\ell}$ where $0 \le \ell \le n$. Then, writing $\Duw_{n \lra \ell}$ for the down-up walk between $X^{(n)}$ and $X^{(\ell)}$ according to $\pi$, $\Udw_{\ell \lra n}$ for the up-down walk on between $X^{(\ell)}$ and $X^{(n)}$ according to $\pi$,  $\Paqx_{n \lra \ell}$ for the expanderized down-up walk via $H$, and $\Papx$ for the expanderized up-down walk via the graph $H^2$ we have
    \[ \EC\parens*{ \Paqx_{\ell \lra n} } \ge \LS\parens*{ \Udw_{\ell \lra n} } \cdot \Gap(H^2)~~\textrm{and}~~\EC\parens*{\Papx_{\ell \lra n}} \ge \LS\parens*{\Udw_{\ell \lra n}} \cdot \Gap(H^2),\]
    where $\Gap(\bullet)$ is as defined before.
\end{theorem}
Both \cref{thm:gshow} and \cref{thm:eshow} follow from an argument showing that the expanderized up-down walks attain a good approximation for a dampened version of the regular up-down walks,
\begin{theorem}[Simplified Version of \cref{thm:exp-close}]\label{thm:show-close}
    Let $(X, \pi)$ be an $n$-partite simplicial complex and $0 \le \ell \le n$. Writing $\Udw_{\ell \lra n}$ for up-down walk between $X^{(\ell)}$ and $X^{(n)}$  according to $\pi$ and $\Papx$ for the expanderized up-down walk between $X^{(\ell)}$ and $X^{(n)}$ via a degree-regular graph $H$ on $\binom{[n]}{\ell}$, we have
    \[ \norm*{ \Papx_{\ell \lra n} - (1 - \lambda(H)) \cdot \Udw_{\ell \lra n} }_\opp \le \lambda(H),\]
    where $\lambda(H)$ is the \ref{eq:two-sided} of $H$ defined by $\lambda(H) = \max\set*{ \lambda_2(H), \lambda_{\min}(H)|}$.
\end{theorem}
This result is proven using ideas from the derandomized graph squaring algorithm of \cite{RozenmanV05}: One can think of the up-down walk as running a random walk on a clique when picking which vertex to remove while going down\footnote{Step (2) in the presentation above.} -- the expanderized up-down walk replaces this clique with an expander graph and as a result achieves a good approximation of the up-down walk. The statements of \cref{thm:gshow} and \cref{thm:eshow} concerning the expanderized down-up walk are obtained by noticing that the expanderized down-up walk can be decomposed into an expanderized down- and an expanderized up-walk. We use the up-down walk in our proofs as we notice that the type identifies each face of $\omega^{(\ell)} \in X^{(\ell)}$ with a unique subset in $\binom{[n]}{\ell}$, i.e.~the vertex set of $H$. This differs from the analysis of \cite{RozenmanV05} in the following manner.
In the expanderized up-down walk, the face we are at after step (1) and the randomness used in step (2) is not enough to determine which face we arrive at after step (3).
In particular, if we start from different faces before step (1), it is possible to arrive at different faces after step (3), even if the face we are at after step (1) and the randomness used in step (2) is the same.
On the other hand, in a typical derandomized product construction \'a la \cite{RozenmanV05}, this would necessitate arriving at the same vertex.

Whereas \cref{thm:gshow} shows that the loss one suffers in the spectral gap is related to the \ref{eq:gap-def} of $H$ is not surprising, we note that \cref{thm:eshow} shows that one can obtain entropy contraction for the expanderized walks paying only a price for the expansion of $H$ -- and not it's entropy contraction factor $\EC(H)$ which can be considerably worse, especially if $H$ is a sparse graph. This is achieved by the following intermediate inequality, which by appealing to an argument of \cite{Miclo97} can be used to bound the entropy contraction constants,
\begin{theorem}[Simplified Version of \cref{cor:entropicstuff}]
Let $(X, \pi)$ be an $n$-partite simplicial complex and let $H$ be a degree regular graph on the vertex set $\binom{[n]}{\ell}$. Writing $\Udw_{\ell \lra n}$ for the up-down walk between $X^{(\ell)}$ and $X^{(n)}$ according to $\pi$ and $\Papx_{\ell \lra n}$ for the expanderized up-down walk between $X^{(\ell)}$ and $X^{(n)}$ via $H$, we have
    \[ \LS\parens*{ \Papx_{\ell \lra n} } \ge \LS\parens*{ \Udw_{\ell \lra n} } \cdot \Gap(H). \]
\end{theorem}
\subsubsection*{Sampling Using the Expanderized Walks}
In \cref{sec:hijack} we analyze the mixing times of expanderized walks for two different problems (i) sampling $q$-colorings of a graph $G$ with maximum degree $\Delta = O(1)$ where $O(\Delta) \ge q \ge (11/6 - \ee)\Delta$ and (ii) sampling from the Ising model with interaction matrix $\Jay \in \RR^{n \times n}$ and external field $\etch \in \RR^{n \times n}$ under the assumption that $\Jay$ is PSD and satisfies $\norm*{\Jay}_\opp \le 1$. Before presenting our results, we briefly talk about these problems and state of the art sampling results for them.  

We note that lifting rapid mixing results of the usual down-up walk to the setting of expanderized walks often proceeds in black-box fashion, and these examples are chosen as representative results obtained with fundamentally different techniques. In particular, this list of applications is chosen as representative applications and is not meant be exhaustive. For example, one can easily extend our results to prove rapid mixing in the case of $p$-spin models \cite{AnariJKPV23, Lee23} and ferromagnetic Ising/Potts models \cite{BlancaCPCPS22} under the suitable assumptions ensuring bounded marginals.

Let $q \in \NN$ and a graph $G = (V, E)$ on $n := |V|$ vertices be given: a $q$-coloring $\chi: V \to [q]$ of a graph $G = (V, E)$ is an assignment of vertices to colors in $[q]$ such that no two pair of adjacent vertices receive the same color, i.e.~$\chi(u) \ne \chi(v)$ for all $\set*{u, v} \in E$. In the sampling problem, one is interested minimizing the the number of colors in relation to the maximum degree $\Delta$ of this graph -- a natural conjecture is that the down-up walk mixes in time $O( n\log n)$ whenever $q \ge \Delta + 2$. An $O(n \log n)$ mixing time bound for $q > 2\Delta$ was established by \cite{Jerrum95}, which was simplified by \cite{BubleyD97}. Then, an $O(n^2)$ mixing time was proven by \cite{Vigoda00} in the case $q > 11/6 \Delta$. Recently, this bound was improved for to $q > (11/6 - \ee) \Delta$ in the recent work of \cite{ChenDMPP19} where $\ee \approx 10^{-5}$ is a small constant. Finally, an $O(n \log n)$ mixing time was proven by \cite{Liu21, BlancaCPCPS22} under the assumption that $\Delta = O(1)$. We show that expanderized random walks rapidly mix in the setting where \cite{Liu21, BlancaCPCPS22}'s result holds,
\begin{theorem}[Simplified Version of \cref{thm:colhijack}]\label{thm:colshow}
    Let $G$ be a graph on $n$ vertices of maximum degree $\Delta \le O(1)$ and $H_n$ be a labelled graph on $[n]$ of constant  degree and \ref{eq:two-sided} $\lambda(H_n)$ bounded away from 1. Then, for some absolute constant $\ee \approx 10^{-5}$,\footnote{See \cite{ChenDMPP19}} and any $K = O(1)$, if $(11/6 + K)\Delta \ge q \ge (11/6 - \ee) \cdot \Delta$, the mixing time of the expanderized down-up walk $\Paqx_n$ on $q$-colorings of $G$  with auxilary graph $H_n$ satisfies,
    \[\Tmix(\Paqx_n, \ee) \le C_0 \cdot n \parens*{ \log n +  \log \ee^{-1}},\]
    where $C_0$ is a universal constant not depending on $n$ but on $\Delta$.
\end{theorem}
We note that by \cite{Alon21}, for each large enough $n \in \NN$, one can \emph{efficiently} construct a graph $H_n$ that will satisfy the requirements of \cref{thm:colshow}. Our proof of \cref{thm:colshow} almost follows in blackbox fashion assuming \cite{Liu21, BlancaCPCPS22}'s result. Furthermore, we note that our algorithm will require $O(n \log n)$ random bits in total -- as in every step we require $O(1)$ bits to sample the vertex whose color we update and $O(1)$ bits to sample a new color for this vertex. The usual down-up walk in contrast, would need at least $\Omega(n \log^2 n)$ random bits, as in every step the cost of sampling the vertex whose color we update is $\Omega(\log n)$ bits.

The Ising model $\mu := \mu_{\Jay, \etch}$ with interaction matrix $\Jay \in \RR^{n \times n}$ and external field $\etch \in \RR^n$ is  a probability distribution on the hypercube $\set*{+1, - 1}^n$ which assigns to each $\eks \in \set*{+1, -1}^n$ the measure,
\[ \mu(\eks) = \frac{ \exp\parens*{ \frac{1}{2} \cdot \inpr*{ \eks, \Jay \eks} + \inpr*{ \etch, \eks} }}{Z(\Jay, \etch)}~~\textrm{where}~~Z(\Jay, \etch) = \sum_{\eks \in \set*{+1, -1}^n} \exp\parens*{ \frac{1}{2} \cdot \inpr*{ \eks, \Jay \eks} + \inpr*{ \etch, \eks} }. \]
Quite recently, \cite{EldanKZ22} established that the \ref{eq:gap-def} of the down-up walk on the hypercube according to $\mu$ is at least $\frac{1 - \norm*{\Jay}_\opp}{n}$ when $\Jay$ is PSD and has small operator norm, i.e.~$\norm*{\Jay}_\opp < 1$. In particular, this implies a mixing time of $O\parens*{\frac{n}{1 - \norm*{\Jay}_\opp} (n + \norm*{\etch}_{\ell_1})}$. Subsequently, this \ref{eq:gap-def} bound was promoted to a bound on the modified-log Sobolev constant by \cite{AnariJKP22}. By employing  a clever argument based on the \emph{approximate exchange property} \cite{AnariLOVV21}, they established a mixing time bound of $O(n \log n)$ -- bypassing the dependence on the external field $\etch$ completely. 

We show that in this setting, the expanderized walks mix rapidly so long as the maximum entry $\norm*{\etch}_{\ell_\infty}$ does not depend on $n$.

\begin{theorem}[Simplified Version of \cref{thm:ishijack}]
    Let $(X^{(\Jay, \etch)}, \mu_{\Jay, \etch})$ be the simplicial complex defined above corresponding to the Ising model defined by the interaction matrix $\Jay \in \RR^{n \times n}$ and external field $\etch \in \RR^n$ and $H_n$ a constant degree graph whose \ref{eq:two-sided} is a constant -- bounded away from 1. Under the assumption that $\Jay$ is PSD and satisfies $\norm*{\Jay}_\opp \le 1$, the following bound holds for the expanderized down-up walk with the auxillary graph $H_n$,
\[ \Tmix(\Paqx_n, \ee) \le \frac{ O\parens*{\norm*{\etch}_{\ell_\infty}} \cdot n}{(1 - \norm*{\Jay}_\opp)^2} \parens*{\log( n  + \norm*{\etch}_{\ell_1} )+ \log \ee^{-1} },\]

where the  $O(\bullet)$ notation hides a universal constant not depending on $n$ or $\Jay$.
\end{theorem}
The $\norm*{\etch}_{\ell_\infty}$ blow-up in the mixing time is due to the reliance of our entropy contraction bounds on the log-Sobolev constant of the up-down walk. When $\norm*{\etch}_{\ell_\infty} = O(1)$ and $\norm*{\Jay}_\opp \ll 1$, we note that -- ignoring numerical difficulties in simulating biased coins -- this random walk requires $O(n \log n)$ random bits as the walk in \cref{thm:colshow}. In this regime the standard Glauber dynamics still requires $O(n \log^2 n)$ random bits. In case where every 2 by 2 principle submatrix of $\Jay$ has operator norm at most $\theta$, we can replace the $(1 - \norm*{\Jay}_\opp)^2$ in the denominator by $(1- \norm*{\Jay}_\opp) \cdot (1 - \theta)$.
\subsubsection*{Functional Inequalities on Simplicial Complexes}
We also prove results which will provide us with the tools necessary to take advantage of \cref{thm:eshow} in settings of interest to the random sampling community. Our main tool here is the Garland method \cite{Garland73} -- more generally the local-to-global method. The improvements we make to the state of the art results are modest at best and in most cases can be recaptured by the local-to-global method of \cite{AlimohammadiASV21}. However, since our proofs are very simple and to the best of our knowledge contain some results which have not explicitly appeared in the literature before (such as \cref{thm:mixing-show}), we choose to include them in the present paper.

We recall that given a distribution $\pi: \Omega \to \RRpp$, and a convex function $\Phi:\RRp \to \RRp$, the \ref{eq:pent-defn} functional $\Ent_{\pi}^{\Phi}(\eff)$ is defined as follows for all $\eff \in \RRp^{\Omega}$, 
\[ \Ent_\pi^{\Phi}(\eff) = \Exp_{\omega \sim \pi} \Phi(\eff(\omega)) - \Phi\parens*{ \Exp_{\omega \sim \pi} \eff(\omega) }.\]
We say that the random walk $\Pii \in \RR^{\Omega_1 \times \Omega_2}$ with two distinguished measures $\pi_1: \Omega_1 \to \RRpp$ and $\pi_2: \Omega \to \RRpp$ satisfying $\pi_1 \Pii_1 = \pi_2$ is said to satisfy a \ref{eq:pent-con} inequality with constant $C > 0$ if the following equality holds for all $\eff \in \RR^{\Omega_2}$,
\[ \Ent_{\pi_1}\parens*{ \Pii \eff} \le (1 - C) \cdot \Ent_{\pi_2}\parens*{ \eff }.\]
The largest constant $C$ for which such an inequality holds is called the \ref{eq:pent-con} constant of $\Pii$ and is denoted by $\CF_\Phi(\Pii)$. Given a simplicial complex $(X, \pi)$ of rank $n$ -- not necessarily partite -- we will prove several results concerning the \ref{eq:pent-con} of the so-called \ref{eq:upw-def} $\Upo_{\ell \to n} \in \RR^{X^{(\ell)} \times X^{(n)}}$ between $X^{(\ell)}$ and $X^{(n)}$. The right action of $\Upo_{\ell \to n}$ on $\eff \in \RR^{X^{(n)}}$ is defined by,
\[ [\Upo_{\ell \to n} \eff](\widehat\omega) = \Exp_{\omega \sim \pi}\sqbr*{ \eff(\omega) \mid \omega \supset \widehat\omega}~~\textrm{for all}~~\widehat\omega \in X^{(\ell)}.\]

 For the choice of $\Phi(t) = t \log t$, the \ref{eq:pent-defn} functi$\Ent_{\pi}^{\Phi}(\bullet)$ corresponds to the usual entropy functional $\Ent_\pi(\bullet)$ defined above and $\CF_\Phi(\Pii)$ is simply $\EC(\Pii)$. For the choice of $\Phi(t) = t^2$, $\Ent^{\Phi}_\pi(\bullet)$ simply corresponds to the \ref{eq:v-def} functional $\Var_\pi(\bullet)$ -- further it is well understood that that $\CF_\Phi(\Upo_{\ell \to n})$ corresponds to the \ref{eq:gap-def} of the walk $\Duw_{n \lra \ell}$. 
 
We note that the concepts of \ref{eq:pent-defn} and \ref{eq:pent-con} are dual concepts to the concept of $f$-divergences and strong data processing inequalities studied in the context of higher order random walks in \cite{AlimohammadiASV21, AnariJKP22a, AnariJKPV23, Lee23}
and more generally in \cite{Chafai04, Raginsky16}. In particular, this duality is the reason why we study the \ref{eq:pent-con} of the \ref{eq:upw-def} whereas the works mentioned above study the $f$-divergence contraction for the \ref{eq:downw-def} $\Doo_{n \to \ell} \in \RR^{X^{(n)} \times X{^(\ell)}}$, whose left action on distributions $\mu: X^{(n)} \to \RRp$ is defined by,
\[ \sqbr*{\mu\Doo_{n \to \ell}}(\widehat\omega) = \sum_{\omega \in X^{(n)},\atop \omega \supset \widehat\omega } \frac{ \mu(\omega)}{\binom{n}{\ell}}~~\textrm{for all}~~\widehat\omega \in X^{(\ell)}. \]

The two concepts coincide for the special class of homogenous \ref{eq:pent-defn}, cf.~\cite{Raginsky16}, which includes the $\Phi$-entropies induced by $\Phi(t) = t^2$ and $\Phi(t) = t\log t$. 

 The local \ref{eq:pent-con} $\lc_\Phi(\varnothing)$ in $(X, \pi)$ is defined as the smallest constant $c > 0$ satisfying the following inequality, 
\[ \Ent_{\pi_1}\parens*{ \Upo_{1 \to n} \eff } \le c \cdot \Ent_{\pi}( \eff )~~\textrm{for all}~~\eff \in \RRp^{\Omega}, \]
where $\pi_1$ is the marginal distribution on $X^{(1)}$ obtained by,
\[ \pi_1(x) = \frac{1}{n} \cdot \Pr_{\omega \sim \pi}[ x \in \omega ]~~\textrm{for all}~~x \in X^{(1)}.\]
Similarly, we define the constant $\lec(\widehat\omega)$ in $(X, \pi)$ analogously by passing to the simplicial complex $(X_{\widehat\omega}, \pi^{(\widehat\omega)})$, where we set
\[ X_{\widehat\omega} = \set*{ \omega  \setminus \widehat\omega \mid \omega \supset \widehat\omega, \omega \in X},\]
and define the distribution $\pi^{(\widehat\omega)}: X_{\omega}^{(n - |\widehat\omega|)} \to \RRpp$ by,
\[ \pi^{(\widehat\omega)}(\alpha) = \Pr_{\omega \sim \pi}\sqbr*{ \omega = \widehat\omega \sqcup \alpha \mid \omega \supset \widehat\omega}.\]
We show,
\begin{theorem}[Simplified Version of \cref{thm:pent-c}]\label{thm:pentc-show}
    Let $(X, \pi)$ be a simplicial complex of rank $n$ and $\Phi: \RRp \to \RRp$ a convex function. We write $\lc^{(i)}_\Phi(X, \pi) = \max_{\widehat\omega \in X^{(i)}} \lc_\Phi(\widehat\omega)$. Then,
    \[ \CF_\Phi\parens*{\Upo_{\ell \lra n}} \ge \prod_{j = 0}^{\ell - 1} \parens*{ 1 - \lc^{(j)}_\Phi(X, \pi)}.\]
\end{theorem}
Our proof of \cref{thm:pentc-show} is a simple instantiation of the Garland method and is inspired by the exposition in \cite{ChenE22} combined with the chain rule for \ref{eq:pent-defn}. After submitting our results to arxiv, it came to our attention that the same proof technique for proving \cref{thm:pentc-show} already appeared in \cite{Liu23} in the context of variance contraction -- though only for the case $\ell = n-1$.
    
Let $\widehat\omega \in X$ be such that $|\widehat\omega| \le n - 2$. We recall that the \ref{eq:link-def} graph $G_{\widehat\omega} = (X^{(1)}_{\widehat\omega}, X^{(2)}_{\widehat\omega}, c_{\widehat\omega})$ with the vertex set $X_{\widehat\omega}^{(1)}$ and the edge set $X_{\widehat\omega}^{(2)}$ where
\[ \textrm{for all}~~x, y \in X_{\widehat\omega}^{(1)}~~\textrm{we have}~~c_{\widehat\omega}(x, y) = \begin{cases} 0 &~\textrm{if}~x = y,\\
\Pr\sqbr*{ \omega \supset \widehat\omega \sqcup \set*{x, y} \mid \omega \supset \widehat\omega}&~\textrm{otherwise.}\end{cases}\]
We will write $\Emm_{\widehat\omega}$ for the random walk matrix of the graph $G_{\widehat\omega}$ where transitions are taken with probability proportional to the weight function. A direct consequence of \cref{thm:pentc-show} is the following bound on the spectral gap,
\begin{corollary}[Spectral Gap Bound]\label{thm:mixing-show}
    Let $(X, \pi)$ be a simplicial complex of rank $n$ and $0 \le \ell \le n$. Writing $\Gap_k(X, \pi) := \min_{x \in X^{(k)}} \Gap(\Emm_x)$
    we have
    \[ \Gap(\Duw_{n \leftrightarrow \ell}) \ge \frac{n - \ell}{n} \cdot \prod_{i = 0}^{\ell - 1} \Gap_i(X, \pi).\]
\end{corollary}
This result is best compared with the following result due to \cite{AlevL20},
\begin{theorem}
    Let $(X, \pi)$ a simplicial complex of rank $n$. Writing $\Gap_k(X, \pi) := \min_{\widehat\omega \in X^{(k)}} \Gap(\Emm_{\widehat\omega})$ for all $0 \le k \le n-2$,
    we have
    \[ \Gap(\Duw_{n \lra n-1}) \ge \frac{1}{n} \cdot \prod_{i = 0}^{n - 2} \Gap_i(X, \pi).\]
\end{theorem}

Similar results to \cref{thm:mixing-show} for variance contraction was proven in \cite{ChenLV20, GuoM21, StefankovicV23}, however for the case $\ell = n-1$ they do not necessarily
    recover the main result of \cite{AlevL20} whereas \cref{thm:mixing-show} does
    indeed recover this guarantee. 

When $\Phi(t) = t \log t$, we will simply write $\lec(\widehat\omega)$ in place of $\lc_\Phi(\widehat\omega)$. We note that this immediately implies,
\begin{corollary}\label{cor:ramen}
     Let $(X, \pi)$ a simplicial complex of rank $n$ and $0 \le \ell \le n$. Writing $\lec_k(X, \pi) := \min_{\widehat\omega \in X^{(k)}} \lec(\widehat\omega)$ for all $0 \le k \le \ell -1$,
    we have
    \[ \EC(\Upo_{\ell \to n}) \ge \prod_{i = 0}^{\ell - 1} \parens*{1- \lec_i(X, \pi)}.\]
\end{corollary}
We observe that when the distribution $\pi$ is $a$-entropically independent in the sense of \cite{AnariJKP22}, we have $\lec_i(X, \pi) = \frac{1}{a \cdot (n - i)}$ and the resulting lowerbound in  \cref{cor:ramen} is the same as the bound given in \cite{AnariJKP22}.

We also note that we can relate the entropy contraction constant $\EC(\Upo_{\ell \to n})$ to the log-Soblev constant thusly. We call the distribution $\pi$ over $X^{(n)}$ $b$-marginally bounded if we have,
\[ \Pr_{\omega \sim \pi}\sqbr*{x \in \omega \mid \widehat\omega \subset \omega} \ge b~~\textrm{or}~~\Pr_{\omega \sim \pi}\sqbr*{x \in \omega \mid \widehat\omega \subset \omega} = 0\]
for all $x \in X_{\widehat\omega}^{(1)}$.
\begin{lemma}[Simplified Version of \cref{lem:lsi}] \label{thm:showent}
     Let $(X, \pi)$ be a simplicial complex of rank $n$. For any $\widehat\omega \in X$, we set $\Gap(\Emm_{\widehat\omega})$.
    where $\Emm_{\widehat\omega}$ is the \ref{eq:link-def} of $\widehat\omega$ and $\Gap(\bullet)$ denotes the \ref{eq:gap-def}. Then $\Duw_{n \lra \ell}$ for the up-down walk between $X^{(n)}$ and $X^{(\ell)}$ and $\Udw_{n-1}$ for the up-down walk between $X^{(n-1)}$ and $X^{(n)}$ (both according to $\pi$), we have
    
	Recalling that $\EC(\bullet)$ denotes the \ref{eq:pent-con} for $\Phi(t) = t \log t$, we have
    \begin{align*}
		\LS(\Duw_{n \lra \ell})&~\ge~
		C_{b,\ell} \cdot \EC\parens*{\Upo_{\ell \to n}},\\
    \LS(\Udw_{n-1}) &~\ge~ C_b \cdot \Gap_{n-2}(X, \pi) \cdot \EC(\Upo_{n-2\to n-1}),
    \end{align*}
    where $C_{b}$ and $C_{b, \ell}$ are constants depending on $b$ and $n-\ell$ only.
\end{lemma}
\begin{remark}
    Here we are thinking $C_b = O(\log(1/b)^{-1})$ and $C_{b, \ell} = O(( \ell \cdot \log(1/b))^{-1})$, i.e.~the dependence of $C_b$ and $C_{b, \ell}$ on $b$ is inversely logarithmic. Similarly, the dependence of $C_{b, \ell}$ on $\ell$ inversely linear.
\end{remark}
In particular, when the distribution at hand is $b$-marginally bounded for some $b = O(1)$ and $\ell = O(1)$ is a constant, \cref{thm:showent} in conjunction with \cref{thm:eshow} indicates in the entropy contraction of our expanderized random walk we only pay a price according to the \ref{eq:two-sided} of the graph $H_n$ that we use in \cref{thm:eshow}. 

\subsection{Related Work}
High dimensional expansion has proven itself to be a very successful research program for establishing mixing times for down-up walks. For example \cite{KaufmanM17, DinurK17, KaufmanO18, DiksteinDFH18, AlevL20} use spectral local-to-global arguments for establishing spectral gap bounds for these walks. In conjunction with the spectral independence framework, due to \cite{AnariLO20, CGSV20, FGYZ20}, these results paved the way for many new in the field of random sampling: rapid mixing of the down-up walk for the hardcore model in the uniqueness regime \cite{AnariLO20}, rapid mixing of the down-up walk for sampling graph colorings in correlation decay regime \cite{FGYZ20, CGSV20}, optimal mixing for many Markov chains of interest\cite{ChenLV21, BlancaCPCPS22, Liu21}. For more information regarding spectral independence, we refer the reader to the excellent survey \cite{StefankovicV23} and the references therein. In \cite{AnariJKP22a, AnariJKPV23, ChenLV21, GuoM21} local-to-global strategies for establishing entropic contraction bounds was studied. In \cite{ChenE22} a connection between these local-to-global methods and the stochastic localization framework of \cite{Eldan13} was explored. We refer to the works~\cite{Klartag18, Eldan20, Chen21, KlartagM22, EldanS22, EldanKZ22} and references therein for applications of the stochastic localization framework. Our inductive strategy for establishing \ref{eq:pent-con} on simplicial complexes is heavily inspired by the presentation in \cite{ChenE22}. In \cite{Lee23} mixing estimates about the walk $\Duw_{n \lra n -1}$ is used to obtain estimates for $\Duw_{n \lra \ell}$ for all $\ell < n- 1$. The key intuition behind this work is  the observation that the
down move of the down-up walk is (passively) utilizing an expander, the down-move of the down-up walk of the so-called \emph{Bernoulli-Laplace model}, and that one can use the expansion of this walk to show that once $\ell$ decreases the mixing times estimates get better and better. Morally, this is very similar to our idea of picking the replacement-indices for our expanderized walks via an expander walk as opposed to sampling them uniformly at random. For other classical techniques which can be used to bound mixing times of Markov chains, we refer the reader to the texts \cite{AldousF95, MontenegroT05, WilmerLP09}.

In contrast with down-up walks, results establishing rapid mixing for the random walk $\Psc$ are fewer \cite{DiaconisR00, Hayes06, DyerGJ06, RobertsR15} and mostly rely on estimates on the Dobrushin matrix \cite{Dobrushin70}. \cite{AlevP23} studied the mixing time of this random walk using techniques of high dimensional expansion, however their techniques fell short of establishing mixing time bounds under the assumption of spectral independence.

The work of \cite{FengGWWY23} is also related to our work in spirit. In this work, the authors show that under suitable assumptions a wide array of random walks, including the single site systematic scan $\Psc$ and the down-up walk $\Duw$, can be derandomized, i.e.~they devise efficient deterministic counting algorithms on the basis of rapid mixing results for these chains. It is an interesting question whether one can carefully pick the expander graph $H$, to make this derandomization task more efficient.

As mentioned above our expanderized random walks are heavily inspired by the derandomized squaring algorithm of \cite{RozenmanV05}. This algorithm was initially used to give an alternative and simpler proof of the seminal result of \cite{Reingold08} concerning the derandomization of the complexity class $\mathbf{SL}$ and establishing $\mathbf{SL} = \mathbf{L}$. Concretely, both \cite{Reingold08} and the subsequent work of \cite{RozenmanV05} show the existence of a deterministic logspace algorithm deciding undirected graph connectivity. Since then, the derandomized squaring algorithm has also found other uses in derandomization, e.g.~\cite{MurtaghRSV17, MurtaghRSV21}. We conclude by noting that the inital algorithm of \cite{Reingold08} was based on the zigzag product construction \cite{ReingoldVW00}, which has also inspired research in the field of high dimensional expansion \cite{KarniK20}. For more information on expander graphs, we refer the reader to the excellent survey \cite{HooryLW06}.
\subsection{Organization}
Our results about expanderized walks are to be found in \cref{sec:expander}. Our results about functional inequalities and local-to-global analysis in simplicial complexes are to be found in \cref{sec:fi}. These two sections can be read independently of each other. In \cref{sec:hijack} we give some example instances where expanderized walks mix rapidly by utilizing the results proven in \cref{sec:expander} and \cref{sec:fi}.
\subsection{Open Questions and Future Directions}

\begin{itemize} 
    \item A current limitation of our method for proving optimal, $O(n \log n)$, mixing times for many problems on $n$-vertex graphs is our reliance on the log-Sobolev constant $\LS(\Udw)$ of the up-down walk to bound the entropy contraction $\EC(\Paqx)$ of expanderized walks. This presents a natural blocker to extend our methods beyond cases where the target distributions are marginally bounded. One can alternatively try to bound the entropy contraction directly, however a naive calculation shows that it is difficult to avoid a blow up related to the entropy contraction $\EC(H)$ of the graph $H$ here. It is a natural question whether a more cunning analysis, not relying on the log-Sobolev constant, can show that the loss in the entropy contraction one will suffer when passing from the usual walks to the expanderized walks depends only on the \ref{eq:two-sided} of $H$.
    \item Our \cref{thm:show-close} shows that the expanderized up-down walks are close to the regular up-down walks. Can we use this result or a result of similar flavor to establish the hypercontractivity of an expanderized noise operator over a simplicial complex? We recall that for $\rho \in [0,1]$ the usual noise operator $\Tee_\rho \in \RR^{X^{(n)}\times X^{(n)}}$ is defined by the equation
    \[ \Tee_\rho = \sum_{j = 0}^n \rho^j (1 - \rho)^{n -j} \cdot \binom{n}{j} \cdot \Duw_{n \lra j}.\]
    In \cite{BafnaHKL22, BafnaHKL22a, GurLL22} the hypercontractivity of the noise operator was established in various cases of interest. An expanderized hypercontractive noise operator a simplicial complex can be useful in constructing sparser integrality examples for many problems of interest. 
\end{itemize}

\subsection*{Acknowledgements}
VLA and SR would like to thank Fernando Granha Jeronimo for many insightful discussions concerning expander graphs. VLA was supported by the ERC grant of Alex Lubotzky (European Union's
Horizon 
2020/882751), the ISF grant 2669/21 and ERC grant 834735 of Gil Kalai, and ISF grant 2990/21 of Ori Parzanchevski. 
SR was supported by the National Science Foundation under Grant No. 2348489.
\section{Preliminaries}\label{sec:prelim}
\subsection{Linear Algebra}
We will denote functions and vectors by bold faces, i.e.~$\eff \in \RR^V$. The indicator function of $i \in V$ will be denoted by $\one_i$, i.e.~$\one_i(j) = 0$ for all $j \ne i$ and $\one_i(i) = 1$. For $A \subseteq V$, we will write $\one_A = \sum_{a \in A} \one_a$. We will adopt the convention of using $\pi, \nu, \mu: V \to \RRp$ for various probability distributions over $V$.

Let
$\eff, \gee \in \RR^V$ and a measure $\pi: V \to \RRpp$ be given. We will use the notations $\langle \eff, \gee\rangle_{\pi}$
and $\norm{\eff}_\pi$ to denote the inner-product and the norm with respect to
the distribution $\pi$, i.e.  
\begin{equation}\langle \eff, \gee\rangle_{\pi} = \Exp_{x \sim
	\pi}\eff(x)\gee(x) =  \sum_{x \in V}
\pi(x) \cdot \eff(x)\gee(x)  ~~\textrm{ and }~~ \norm{\eff}_\pi^2 = \langle \eff,
\eff\rangle_{\pi}.\label{eq:inpr-defn}
\end{equation}
Given $\eff, \gee \in \RR^n$ we will write $\inpr*{\eff, \gee}_{\ell_2}$ for the inner-product between $\eff$ and $\gee$ in the counting measure, i.e.~$\inpr*{\eff, \gee}_{\ell_2} = \sum_{i = 1}^n \eff(i) \gee(i)$. We will also write $\norm*{ \eff }_{\ell_1}, \norm*{\eff}_{\ell_2}$, and $\norm*{\eff}_{\ell_\infty}$ for the $\ell_1, \ell_2,$ and $\ell_\infty$ norms of $\eff$ respectively. Formally,
\[ \norm*{\eff}^2_{\ell_2} = \sum_{i = 1}^n \eff(i)^2~~\textrm{;}~~\norm*{\eff}_{\ell_1} = \sum_{i = 1}^n |\eff(i)|~~\textrm{; and ;}~~\norm*{\eff}_{\ell_\infty} = \max_{i \in [n]} |\eff(i)|.\]
\subsubsection*{Matrices and Eigenvalues}
In this section, we will recall some results concerning eigenvalues and
eigenvectors of matrices.

Serif faces will be used to denote matrices, i.e.~$\Aye,\Bee \in \RR^{U \times V}$. We
will call a matrix $\Bee \in \RR^{U \times V}$ \ref{eq:rs-def} if rows of $\Bee$ sum up to 1 and $\Bee$ contains no negative entries. Formally,
\begin{equation}\tag{row stochastic}\label{eq:rs-def}\textrm{for all } u \in U, v \in V~~\Bee(u, v) \ge 0~~\textrm{and}~~\Bee
\one = \one.
\end{equation}
Let $\Bee \in \RR^{U \times V}$ and distributions $\pi_U : U \to \RRpp$ and $\pi_V: V \to \RRpp$ be given. The \ref{eq:adj-defn} $\Bee^*$ of $\Bee$ with respect to the measures $\pi_U$ and $\pi_V$ is the unique matrix which satisfies the following equation,
\begin{equation}\label{eq:adj-defn}\tag{adjoint} \langle \eff, \Bee
\gee\rangle_{\pi_U} = \langle \Bee^* \eff, \gee\rangle_{\pi_V}~~\textrm{ for all
} \eff \in \RR^U, \gee \in \RR^V.
\end{equation}
If $U = V$ and $\pi_U = \pi_V$, the operator
$\Bee$ is called self-adjoint when $\Bee^*= \Bee$.  If $\Bee$ is a
row-stochastic matrix, we will call $\Bee^*$ the time-reversal of $\Bee$ with
respect to $\pi_U,\pi_V$ and
say that $\Bee$ is reversible if $\Bee = \Bee^*$. It is well known that the
operator $\Bee^* \in \RR^{V \times U}$ is uniquely determined by the choice of
$\Bee \in \RR^{U \times V}$ and the
inner-products defined by $\pi_U$ and $\pi_V$ (see
e.g.~\cite[p.~318]{Saloff-Coste97}),
\begin{proposition}\label{prop:adjoint-defn}
    Let $\Bee \in \RR^{U \times V}$ be arbitrary. We write $\Bee^*$ for the
    adjoint operator to $\Bee$ with respect to the inner-products defined by
    the distributions $\pi_U$ and $\pi_V$. 
    Then,
    \[ \Bee^*(y, x) = \Bee(x, y) \cdot \frac{\pi_U(x)}{\pi_V(y)}~~\textrm{for
    all}~~x \in U, y \in V.\]
\end{proposition}

We also recall the following standard fact which is an immediate consequence of
\cref{prop:adjoint-defn},
\begin{fact}\label{fac:reversal}
	If $\Bee \in \RR^{U, V}$ is a row-stochastic matrix satisfying $\pi_U \Bee =
	\pi_V$, then the adjoint matrix $\Bee^*$ with respect to $\pi_U, \pi_V$ is
	also row-stochastic and satisfies $\pi_V \Bee^* = \pi_U$.
\end{fact}

It is well known that a self-adjoint matrix $\Aye \in \RR^{V \times V}$ has $|V|$ real eigenvalues. We will write,
\[ \lambda_1(\Aye) \ge \lambda_2(\Aye) \ge \cdots \ge \lambda_{|V|}(\Aye) := \lambda_{\min}(\Aye),\]
for the sequence of eigenvalues of $\Aye$ sorted in decreasing order. We say that the matrix is positive semi-definite, henceforth PSD, if it is self-adjoint and satisfies $\lambda_{\min}(\Aye) \ge 0$.

Given a matrix $\Aye \in \RR^{V \times V}$ and a distribution $\mu: V \to \RRpp$, we will write $\norm*{ \Aye }_{\opp, \mu}$ for the operator norm of $\Aye$, defined in the following manner
\begin{equation}\label{eq:opn-def}\tag{operator norm} \norm*{ \Aye }_{\opp, \mu} := \max\set*{ \left. \frac{ \norm*{\Aye \eff}_\mu }{\norm{\eff}_\mu} ~\right|~\eff \in \RR^V~\textrm{and}~\eff \ne 0}.
\end{equation}
If $\Aye$ is self-adjoint with respect to the measure $\mu$, we have $\norm*{\Aye}_{\opp, \mu} = \max\set*{\lambda_1(\Aye), |\lambda_{\min}(\Aye)|}$. When $\mu$ is the counting measure, we will simply write $\norm{\bullet}_{\ell_2}$.

Similarly when $\Aye \in \RR^{V \times V}$ is a reverisble row-stochastic matrix, with stationary measure $\mu$. We will write $\lambda(\Aye)$ for the \ref{eq:two-sided} of $\Aye$. Formally,
\begin{equation}\label{eq:two-sided}\tag{two-sided expansion}
    \lambda(\Aye) = \max\set*{ \lambda_2(\Aye), \Abs*{\lambda_{\min}(\Emm) }}.
\end{equation}

When $\Aye$ represents the simple random walk over an undirected graph $H = (V, E)$, i.e.~
\[\Aye(i, j) = \frac{\one[\set*{i, j} \in E]}{\deg(i)}~~\textrm{for all}~~i, j \in V,\]
we will simply write $\lambda(H)$ instead of $\lambda(\Aye)$. For convenience, we recall
\begin{fact}\label{fac:baby}
    Let $H = (V, E)$ be a $k$-regular graph and suppose $\Aye$ represents the random walk over $H$. Then, $\uni_V \Aye = \uni_V$, i.e.~the uniform distribution on $V$ is stationary for $\Aye$.
\end{fact}
We note that there exist infinite families of graphs such that every graph $H$ in the family has constant degree and $\lambda(H)$ bounded above by a constant bounded above by $1$~\cite{LubotzkyPS86, Margulis88}. In this paper, we will consider families that contain graphs on $n$ vertices for \emph{every} sufficiently large $n$. Such constructions were given in~\cite{Alon21}, and in particular were based on the infinite families from~\cite{LubotzkyPS86, Margulis88}. We refer the reader to the excellent survey \cite{HooryLW06} for more information on expander graphs. For our purposes we will only need to rely on the following result,
\begin{theorem}[Theorem 1.1 in \cite{Alon21}, simplified]\label{thm:alon}
    For every prime number $p \equiv 1 \pmod 4$, and every sufficiently large $n > n_0(p)$, there exists a strongly explicit\footnote{i.e.~the list of neighbors of any vertex can be generated deterministically in time $\plog(n)$} construction of a $d$-regular graph $H_n$ on $n$-vertices with $\lambda(H_n) \le \frac{(1 +\sqrt 2)\sqrt{d-1} + o(1)}{d}$, where the $o(1)$ vanishes as $n$ tends to infinity.  
\end{theorem}
We emphasize that if the prime $p \equiv 1 \pmod 4$ is a constant, i.e.~has no dependency on $n$, then the graph $H_n$ is a $d$-regular expander graph with $d = O(1)$ and $\lambda(H_n) = 1 - \ee$ for some constant $\ee := \ee(p)$, i.e.~the \ref{eq:two-sided} of $H_n$ is bounded away from 1.

The following variational characterizations of $\lambda_2(\bullet), \lambda_{\min}(\bullet),$ and $\lambda(\bullet)$ are simple consequences of the Courant-Fischer-Weyl principle and the Perron-Frobenius Theorem, see for example \cite{Bhatia2013, HornJ12, AldousF95}.
\begin{fact}\label{fac:cf-baby}
    Let $\Aye \in \RR^{V \times V}$ be a reversible row-stochastic matrix with respect to the measure $\mu: V \to \RRpp$. Then, the following hold,
    \begin{enumerate}
        \item $\lambda_2(\Aye) = \max\set*{ \left.\frac{ \inpr*{ \eff, \Aye \eff}_\mu }{\norm{\eff}^2_{\mu}}~\right|~ \eff \in \RR^V \setminus \boldsymbol 0, \inpr*{\eff,\one}_\mu = 0},$
        \item $\lambda_{\min}(\Aye) = \min\set*{ \left. \frac{ \inpr*{ \eff, \Aye \eff}_\mu}{\norm{\eff}^2_{\mu}}~\right|~\eff \in \RR^V \setminus \boldsymbol 0, \inpr*{\eff,\one}_\mu = 0},$
        \item $\lambda(\Aye) = \max\set*{ \left.\frac{\norm*{\Aye \eff}_\mu}{\norm{\eff}_\mu}~\right|~\eff \in \RR^V \setminus \boldsymbol 0, \inpr*{\eff, \one}_\mu = 0}$,
        \item $\lambda(\Aye) = \max\set*{ \left.\frac{\Abs*{\inpr*{\eff,\Aye \eff}_\mu}}{\norm{\eff}^2_\mu}~\right|~\eff \in \RR^V \setminus \boldsymbol 0, \inpr*{\eff, \one}_\mu = 0}$.
    \end{enumerate}
\end{fact}
We will also make use of the following simple result,
\begin{fact}\label{fac:switcheroo}
    Let a matrix $\Aye \in \RR^{U \times V}$ and measures $\mu_U: U \to \RRpp$ and $\mu_V: V \to \RRpp$ be given, such that $\mu_U \Aye = \mu_V$. Assume without loss of generality that $|U| \le |V|$, then
    \[ \lambda_j(\Aye \Aye^*) = \lambda_j(\Aye^* \Aye)~~\textrm{for all}~~j=1, \ldots, |U|,\]
    where $\Aye^*$ is the \ref{eq:adj-defn} of $\Aye$ with respect to the measures $\mu_U$ and $\mu_V$.
\end{fact}
\subsection{Probability Distributions}
We will denote the \ref{eq:simplex-defn} with vertices $\Omega$ by $\triangle_\Omega$, i.e.~
\begin{equation}\label{eq:simplex-defn}\tag{probability simplex} 
\triangle_{\Omega} 
= \set*{ \left. \mu: \Omega \to \RRp~\right|~\sum_{\omega \in \Omega} \mu(\omega) = 1}
\end{equation}
 Throughout the paper, we will assume
$\Omega$ (or $X^{(n)}$) to  be a set of $n$-tuples for some $n \ge 1$. Given a set $S \subset [n]$,  the \ref{eq:proj-defn} of $\Omega$ on $S$ is denoted by $\Omega[S]$ , i.e.~
\begin{equation}\label{eq:proj-defn}\tag{projection}
\Omega[S] = \set*{ (\omega_s)_{s \in S} : (\omega_1, \ldots, \omega_n) \in \Omega}.
\end{equation}

For $\omega_S \in \Omega_S$, the notations $\Omega_{\omega_S}$ and
$\mu^{(\omega_S)}$ will be used for the
\ref{eq:p-def} of $\Omega$ and $\mu_S$ respectively, where
\begin{equation}\Omega_{\omega_S} = \set*{\bar\omega \in \Omega[S^c] : \omega_S \oplus
\bar\omega \in \Omega}~~\textrm{and}~~\mu^{(\omega_S)}(\bar\omega) =
\frac{\mu(\omega_S \oplus \bar\omega)}{\sum_{\tilde \omega \in \Omega[S^c]}
\mu(\omega_S \oplus \tilde\omega)},\label{eq:p-def}\tag{$\omega_S$-pinning}
\end{equation}

We recall that the \ref{eq:tv-def} $\norm{\mu - \nu}_{\TV}$ between two distributions $\mu, \nu \in \triangle_{\Omega}$ is defined as follows,
\begin{equation}\label{eq:tv-def}\tag{total variation distance}
\norm*{\mu - \nu}_{\TV} = \frac{1}{2} \cdot \sum_{\omega \in \Omega} \Abs*{ \mu(\omega) - \nu(\omega)}
\end{equation}

Finally, we talk about some conventions that we will use throughout the paper: (i) We will be using the notation $\uni_{A} \in \triangle_{A}$ to denote the uniform distribution over various finite sets $A$. (ii) When we want to emphasize that the a distribution $\mu \in \triangle_{\Omega}$ has full support, we will simply write $\mu: \Omega \to \RRpp$ and emphasize in words that $\mu$ is a distribution as opposed to writing $\mu \in \triangle_{\Omega}.$ 

Finally we recall that the product distribution $\mu \tensor \nu \in \triangle_{\Omega \times \Omega'}$, given $\mu \in \triangle_{\Omega}$ and $\nu \in \triangle_{\Omega'}$ is defined by setting
\[ (\mu \tensor \nu)(\omega, \omega') = \mu(\omega) \cdot \nu(\omega')~~\textrm{for all}~~\omega \in \Omega, \omega' \in \Omega.\]

\subsection{Functional Inequalities, Isoperimetric Constants, and Mixing Times}\label{ss:fuuuun}
Given a distribution $\mu \in \triangle_{\Omega}$ and a convex function $\Phi: \RRp \to \RRp$ \ref{eq:pent-defn} functional $\Ent^\Phi_\mu(\bullet)$ is defined by the equation,
\begin{equation}\label{eq:pent-defn}\tag{$\Phi$-entropy}
    \Ent^\Phi_\mu(\eff) = \Exp_{\omega \sim \mu}\Phi\parens*{\eff(\omega)} - \Phi\parens*{ \Exp_{\omega \sim \pi_n} \eff(\omega) }
\end{equation}
for all $\eff \in \RRp^{\Omega}.$

We also recall that for the special choices of $\Phi(t) = t \log t$ and $\Phi(t) = t^2$, the \ref{eq:pent-defn} equals the \ref{eq:v-def} functional $\Var_\mu(\bullet)$ and \ref{eq:ent-def} functional $\Ent_\mu(\bullet)$ respectively. 
\begin{align}
    \Ent_\mu(\eff) &~=~\Exp_{\omega \sim \mu}\sqbr*{ \eff(\omega) \log \eff(\omega)} - \parens*{ \Exp_{\omega \sim \mu} \eff(\omega)} \log\parens*{ \Exp_{\omega \sim \mu}\eff(\omega)},\label{eq:ent-def}\tag{entropy}\\
    \Var_\mu(\eff) &~=~\Exp_{\omega \sim \mu} \eff(\omega)^2 - \parens*{ \Exp_{\omega \sim \mu} \eff(\omega)}^2. \label{eq:v-def}\tag{variance}
\end{align}
\begin{remark}\label{rem:aoe}We notice that when $\eff = c \cdot \one$ for some constant $c \in \RRp$, we have $\Ent_\mu^{\Phi}(\eff) = 0$.
\end{remark}
Let $\Pii \in \RR^{\Omega \times \Omega}$ be a reversible Markov chain, with
stationary measure of $\pi$. A \ref{eq:pcare-def} for $\Pii$ is
an inequality of the form,
\begin{equation}\label{eq:pcare-def}
	C \cdot \Var_\pi(\eff) \le \inpr*{ \eff, (\Ide - \Pii)
	\eff}_\pi~~\textrm{for all}~~\eff \in \RR^{\Omega}.\tag{Poincar\'e
inequality}
\end{equation}
The largest constant $C >0$ for which this inequality holds, is called the
Poincar\'e constant or the \ref{eq:gap-def} of $\Pii$ and
is denoted by $\Gap(\Pii)$. This nomenclature is due to the following
well-known consequence of the Courant-Fischer-Weyl Principle,
\begin{equation}\label{eq:gap-def}\tag{spectral gap}
	\Gap(\Pii) = \min\set*{ \left. \frac{ \inpr*{\eff, (\Ide - \Pii) \eff}_\pi
	}{\Var_\pi(\eff) } ~\right|~ \Var_\pi(\eff) \ne 0} = 1 -
	\lambda_2(\Pii).
\end{equation}

The modified log-Sobolev (\ref{eq:mlsi-def}) and the log-Sobolev (\ref{eq:lsi-def}) inequalities
for a reversible random walk $\Pii \in \RR^{\Omega \times \Omega}$ with
stationary measure $\pi$ are defined to be,
\begin{align}
	C_0 \cdot \Ent_\pi(\eff) &~\le~\inpr*{ \eff, (\Ide - \Pii) \log \eff
	}_\pi&&\textrm{for all}~~\eff \in \RR_{\ge
0}^{\Omega}\tag{mLSI}\label{eq:mlsi-def},\\
	C_1 \cdot \Ent_\pi(\eff^2) &~\le~\inpr*{ \eff, (\Ide - \Pii) \eff}_\pi
							   &&\textrm{for all}~~\eff \in \RR_{\ge
							   0}^{\Omega}\tag{LSI}\label{eq:lsi-def}.
\end{align}
The largest constants $C_0, C_1 \ge 0$ for which \ref{eq:mlsi-def} and
\ref{eq:lsi-def} hold are called the modified log-Sobolev
constant and
the log-Sobolev constant of $\Pii$ respectively -- and they are denoted by
$\MLS(\Pii)$ and $\LS(\Pii)$. Formally,
\begin{align}
	\MLS(\Pii) &~=~\inf\set*{ \left. \frac{ \inpr*{ \eff, (\Ide - \Pii) \log \eff}_\pi
	}{\Ent_\pi(\eff)} ~\right|~ \Ent_\pi(\eff) \ne 0, \eff \in \RR_{\ge 0}^{\Omega}},\label{eq:mlsc-def}\\
	\LS(\Pii) &~=~\inf\set*{ \left.\frac{ \inpr*{ \eff, (\Ide - \Pii) \eff
			}_\pi } {\Ent_\pi(\eff^2)}
			~\right|~
	\Ent_\pi(\eff) \ne 0, \eff \in \RR_{\ge 0}^{\Omega}}.\label{eq:lsc-def}
\end{align}

\begin{fact}[\cite{DiaconisSC96}]\label{fac:evb}
    Let $\pi: \Omega \to \RR_{> 0}$ be a probability distribution and write $\Jay_\pi = \one \cdot \pi$, i.e.~$\Jay_\pi$ is the walk with stationary measure $\pi$ which mixes in a single step.

    Then,
    $\LS(\Jay_\pi) \ge \frac{1 - 2 \pi_{\star} }{\log(\pi_{
	\star}^{-1} - 1)}$ if $|\supp(\pi)| > 2$ else $\LS(\Jay_\pi) = 1$. More generally for any reversible Markov chain $\Emm \in \RR^{\Omega \times \Omega}$ and stationary distribution $\pi$, we have
    $\LS(\Emm) \ge \frac{1 - 2 \pi_{\star} }{\log(\pi_{
	\star}^{-1} - 1)} \cdot \Gap(\Emm)$ if $|\supp(\pi)| > 2$ else $\LS(\Jay_\pi) = \Gap(\Emm)$.
\end{fact}

% We also recall the following beautiful result of \cite{SalezTY23}, showing that bounds on $\MLS(\Pii)$ can be translated to bounds on $\LS(\Pii)$,
%     \begin{theorem}[\cite{SalezTY23}]\label{thm:compcomp}
% 	Let $\Pii \in \RR^{\Omega \times \Omega}$ be a reversible Markov chain with
% 	stationary measure $\pi$. The following inequality holds ,
% 	\[20 \cdot \LS(\Pii) \ge \MLS(\Pii) \cdot \log(1/p)\] where $p =
% 			\min\set*{ \Pii(x, y) : \Pii(x, y) > 0 }$.
% \end{theorem}

For a convex function $\Phi: \RRp \to \RRp$, we also define the \ref{eq:pent-con} constant $\CF_\Phi(\Pii)$ of a Markov chain $\Pii \in \RR^{\Omega_1 \times \Omega_2}$ satisfying $\pi_1 \Pii = \pi_2$ for some choice of measures $\pi_1 \in \triangle_{\Omega_1}, \pi_2 \in \triangle_{\Omega_2}$, as the solution to the following variational problem,
\begin{equation}\label{eq:pent-con}\tag{$\Phi$-entropy contraction}
    \CF_\Phi(\Pii) = 1 - \sup\set*{\left. \frac{\Ent^\Phi_{\pi_1}(\Pii \eff)}{\Ent^\Phi_{\pi_2}(\eff)}~\right|
~\eff \in \RRp^{\Omega}, \Ent^\Phi_{\pi_2}(\eff) \ne 0}.
\end{equation}
We note that $\CF_\Phi(\Pii)$ cruicially depends on the choice of distributions $\pi_1, \pi_2$. Since for our purposes the choice of measures $\pi_1$ and $\pi_2$ will always be clear, we will supress this dependency. 

It is equivalent to define $\CF_\Phi(\Pii)$ as the largest constant $C \in \RRp$ such that the inequality,
\[ \Ent^\Phi_\pi(\Pii \eff) \le (1 - C) \cdot \Ent^\Phi_\pi(\eff),\]
is valid for each $\eff \in \RRp$. When $\Phi(t) = t \log t$, we will simply write $\EC(\Pii)$ in place of $\CF_\Phi(\Pii)$. Similarly, for the choice of $\Phi(t) = t^2$, it is easy to observe that $\CF_\Phi(\Pii) = \Gap(\Pii^*\Pii)$.

\begin{remark}\label{rem:divent}
    \ref{eq:pent-con} and the quantity $\CF_\Phi(\Pii)$ is closely related to the concept of $f$-divergence contraction, studied in the context of higher-order random walks in the works of \cite{AnariJKP22, AnariJKP22a, Lee23}. We note that for the choices of $\Phi(t) = f(t) = t^2$ and $\Phi(t) = f(t) = t \cdot \log t$, $f$-divergence contraction for the walk $\Pii$ is equivalent to \ref{eq:pent-con} $\Pii^*$.  This equivalence holds for the more general class of homogenous entropies associated with $\Phi(t) = f(t)$ \cite{Raginsky16}, i.e.~convex functions $\Phi: \RRp \to \RRp$, for which there is a function $\kappa: \RRp \to \RRp$ that satisfies,
    \[ \Ent^\Phi_{\pi_1}(c \cdot \eff) = \kappa(c) \cdot \Ent^\Phi_{\pi_2}(\eff),\]
    for each $\eff \in \RRp^{\Omega}$. For more information on \ref{eq:pent-con} and $f$-divergence contraction we refer the reader to \cite{Chafai04, Raginsky16, BoucheronLM13}, and the references therein.
\end{remark}

We will also need the following consequence of Jensen's inequality.

\begin{fact}[Data Processing Inequality]\label{fac:dpi} Let $\Pii \in \RR^{\Omega_1 \times \Omega_2}$ be a row-stochastic matrix, satisfying $\pi_1 \Pii = \pi_2$ for probability distributions $\pi_1: \Omega_1 \to \RRpp$ and $\pi_2: \Omega_2 \to \RRpp$. Then, for any convex function $\Phi: \RRp \to \RRp$, we have2
\[ \Ent_{\pi_1}^{\Phi}\parens*{ \Pii \eff} \le \Ent^\Phi_{\pi_2}\parens*{ \eff }~~\textrm{for all}~~\eff \in \RRp^{\Omega_2}.\]
In particular, for any $\Quu \in \RR^{\Omega_2 \times \Omega_3}$ with a measure $\pi_3: \Omega_3 \to \RRpp$ such that $\pi_2\Quu = \pi_3$, we have $\CF_\Phi(\Pii\Quu) \ge \CF_\Phi(\Pii)$ and $\CF_\Phi(\Pii\Quu) \ge \CF_\Phi(\Quu).$
\end{fact}
We provide a proof for \cref{fac:dpi} in \cref{ap:dpi}.

We also recall the following relations between the constants $\MLS(\bullet), \LS(\bullet)$, and $\EC(\bullet)$.  
\begin{lemma}[Lemma 16, \cite{AnariJKP22}]\label{lem:mlsfun}
Let $\Pii \in \RR^{\Omega \times \Omega}$ be a reversible row-stochastic matrix with stationary distribution $\pi$. Then,
\[ \MLS(\Pii) \ge \EC(\Pii)\]
\end{lemma}
\begin{remark}
    \cref{lem:mlsfun} can be generalized to show that contraction results for the \ref{eq:pent-defn} can be utilized to obtain functional inequalities called $\Phi$-Sobolev inequalities. We refer to \cite[Section 14]{BoucheronLM13} and \cite[Section 4]{Raginsky16} for more information on this topic.
\end{remark}
 \begin{lemma}[Proposition 6, \cite{Miclo97}]\label{lem:lsfun}
     Let $\Pii \in \RR^{\Omega_1 \times \Omega_2}$ satisfying $\mu_1 \Pii = \mu_2$, for distributions $\mu_1: \Omega_1 \to \RRpp$ and $\mu_2: \Omega \to \RRpp$. We have,
     \[ \EC(\Pii) \ge \LS(\Pii^* \Pii).\]
 \end{lemma}
 
 \begin{remark}
     The statements of \cref{lem:mlsfun} and \cref{lem:lsfun} have minor cosmetic changes from the statements of \cite[Lemma 16]{AnariJKP22} and \cite[Lemma 6]{Miclo97} respectively.

     The statement in \cite[Lemma 16]{AnariJKP22} considers the contraction of divergences, rather than contraction of entropies. These are equivalent by \cref{rem:divent}.
     Additionally, their definition of the modified log-Sobolev inequality (\ref{eq:mlsi-def}) differs from ours by a factor of 2.

     The statement in \cite[Lemma 6]{Miclo97} assumes square operators. However, as observed in \cite[Lemma 5.11]{ChenLV21} the proof can be extended to rectangular operators also. For completeness, we provide a proof in \cref{ap:lsfun} -- we use their proof as a template and make the minimal syntactic changes necessary.
 \end{remark}

The \ref{eq:epsmix} $\Tmix(\Pii, \ee)$ of the random walk is the least time point $t \in \NN$, such that the distribution $\mu^{(t)} = \mu^{(0)} \Pii^t$ of the random walk $\Pii$ is guaranteed to be $\ee$-close to the stationary distribution $\pi$ in the \ref{eq:tv-def} regardless of the initial distribution $\mu^{(0)}$. In particular,
\begin{equation}\label{eq:epsmix}\tag{$\ee$-mixing time}
    \Tmix(\Pii, \ee) = \min\set*{\left. t \in \NN~\right|~ \norm*{\mu^{(t)} - \pi}_{\TV} \le \ee~~\textrm{for all}~~\mu^{(0)} \in \triangle_{\Omega}}
\end{equation}
It is well known that the functional inequalities and the corresponding isoperimetric constants introduced previously can be used to bound mixing times. We recall two of these results, 
\begin{theorem}[\cite{BlancaCPSV21}]\label{thm:entmix}
    There exists a universal constant $C$ such that, for any reversible random walk $\Pii \in \RR^{\Omega \times \Omega}$ with stationary distribution $\pi: \Omega \to \RRpp$, i.e.~$\pi \Pii = \pi$. We have
    \[ \Tmix(\Pii, \ee) \le \frac{C}{\EC(\Pii)} \cdot  \parens*{{\log\log\frac{1}{\min_{\omega \in \Omega} \pi(\omega)} + \log \ee^{-1}}},\]
    where the constant $C$ does not depend on the pair $(\Pii, \pi)$ and $\EC(\Pii)$ denotes the entropy contraction constant of the random walk $\Pii$.
\end{theorem}
The following result is also well-known, see e.g.~\cite{MontenegroT05}
\begin{theorem}
    Let $\Pii \in \RR^{\Omega \times \Omega}$ be a reversible random walk matrix with stationary distribution $\pi: \Omega \to \RRpp$, i.e.~$\pi\Pii = \pi$. We have,
    \[ \Tmix(\Pii, \ee) \le \frac{1}{1 - \lambda(\Pii)} \cdot \log\frac{1}{\ee \cdot \sqrt{\min_{\omega \in \Omega} \pi(\omega)}},\]
    where $\lambda(\Pii)$ is the \ref{eq:two-sided} of $\Pii$. If $\Pii$ is PSD, then $1 - \lambda(\Pii) = \Gap(\Pii)$.  
\end{theorem}

\subsection{(Partite) Simplicial Complexes}\label{sec:pc}
A simplicial complex is a downward closed collection of subsets of a finite set
$U$. Formally, $X \subset 2^U$ and whenever $\beta \in X$ for all $\alpha \subset \beta$ we have
$\alpha \in X$. The rank of a face $\alpha$ is $|\alpha|$. Given some $j$, we
will adopt the notation $X^{(j)}$ to refer to the collection faces of $X$ of 
rank $j$ and the notation $X^{(\le j)}$ to refer to the collection of faces of $X$
of rank at most $j$, i.e.~
\[ X^{(j)} := X \cap \binom{U}{j}~~\textrm{and}~~X^{(\le j)} := \bigcup_{i = 0}^j
X^{(i)}.\]
We say $X$ is a simplicial complex of rank $n$ if the largest rank of
any face $\alpha \in X$ is $n$. We note
that by definition $X^{(0)} = \set{\varnothing}$.

We say that a simplicial complex $X$ of rank $n$ is pure, if any face $\alpha \in
X^{(j)}$ for any $j < n$ 
is contained in another face $\beta \in X^{(n)}$. Equivalently, in a pure simplicial complex the only inclusion maximal
faces are those of maximal rank. In this article, we will only deal with pure simplicial complexes.

A rank-$n$ pure simplicial complex $X$ is called $n$-partite if we
can partition $X^{(1)}$ into disjoint sets $X[1], \ldots, X[n]$ such that
\[ \textrm{for all}~\beta \in X^{(n)}~\textrm{and for all}~i=1, \ldots,
	n~~\textrm{we have}~~|\beta
\cap X[i]| = 1. \tag{$n$-partiteness}\]
We will call the sets $X[1], \ldots, X[n]$ the sides of the complex $X$. Equivalently, every element of a rank-$n$ face $\beta \in X[n]$
comes from a distinct side $X[i]$. We observe that a bipartite graph is a
2-partite simplicial complex. 
\begin{remark}
	Our notation differs slightly from the preceding work, the notation $X(j)$
	in the preceding work is often used to denote what we have called
	$X^{(j+1)}$, i.e.~$X(j)$ is the set of $j$-dimensional faces. We have avoided this notation to 
	(i) prevent potential confusion with $X[j]$ and (ii) to avoid refering to
	an $n$-partite complex as an $(n-1)$-dimensional $n$-partite complex, as
	was done in e.g.~\cite{Oppenheim18b, DiksteinD19}.
\end{remark}
To keep our nomenclature simple, we will simply
refer to a pure $n$-partite simplicial complex of rank $n$ as an  $n$-partite simplicial
complex, i.e.~we will not consider $n$-partite complexes which are not pure.

For a face $\alpha \in X$ we introduce the notation, $\typ(\alpha) = \set*{ i
\in [n] : \alpha \cap X[i] \ne \varnothing}$ for the type of the face $\alpha$, i.e.~the collection of sides of $X$ that $\alpha$ intersects.

For any $i \in [n]$ and $\beta \in X^{(n)}$ we will write $\beta_i \in X^{(1)}$ for
the unique element of $\beta$ satisfying $\set{\beta_i} = \beta \cap X[i]$. We
will refer to $\beta_i$ as the $i$-th coordinate of $\beta$. We
will also write $\beta_T = \set*{\beta_t : t \in T}$ for all $T \subset [n]$.
We extend this notation to arbitrary faces $\alpha \in X$ and $T
\subset \typ(\alpha)$. In keeping with the view that a face $\alpha \in X$ with
$\typ(\alpha) = \set*{t_1, \ldots, t_k}$ can be represented as a tuple
$(a_{t_1}, \cdots, a_{t_k})$, we will favour the notation $\alpha \oplus
\alpha'$ to denote the union of two faces $\alpha, \alpha' \in X$ with $\typ(\alpha)
\cap \typ(\alpha') = \varnothing$ over the usual notation $\alpha \cup \alpha'$.

We observe that for facets $\beta \in X^{(n)}$, i.e.~faces of maximal rank, we
have $\typ(\beta) = [n]$. Given, $\alpha \in X$ we recall that the link $X_\alpha$ is
defined as,
\[ X_\alpha = \set{ (\beta \setminus \alpha) \in X : \beta \in X, \beta \supset \alpha}.\]
The following observation is immediate,
\begin{fact}
    Let $X$ be an $n$-partite simplicial complex with
	sides $X[1], \ldots, X[n]$ and $\alpha \in X^{(j)}$ for some $j \in [0,n]$. Then, 
    the simplicial complex $X_\alpha$ is an $(n - j)$-partite
	simplicial complex with sides $X_\alpha[j] := X[j] \cap
	X_\alpha^{(1)}$ for 
	$j \in [n]\setminus \typ(\alpha)$.
\end{fact}

For $T \subset [n]$, we will also introduce the notation $X[T]$ to refer to
all faces of $X$ of type $T$, i.e.~
\[ X[T] = \set*{ \alpha \in X : \typ(\alpha) = T}.\]
Notice,
\[ X^{(n)} = X[1,\ldots, n]~~\textrm{and}~~X^{(j)} = \bigcup_{T \in
\binom{[n]}{j}} X[T].\]
\subsubsection*{Weighted Simplicial Complexes}

A weighted simplicial complex $(X, \pi)$ of rank $n$ is a pure simplicial
complex of rank $n$ where $\pi
:= \pi_n$ is a probablity distribution on $X^{(n)}$ with full support, i.e.~$\pi \in
\triangle_{X^{(n)}}$ and $\supp(\pi) = X^{(n)}$.
For $j \in [0, n-1]$, we inductively define the probability distributions
$\pi_j: X^{(j)} \to \RR$ as
\begin{equation}
	\pi_j(\alpha) = \frac{1}{j+1} \sum_{\beta \in X[j+1],\atop \beta \subset
    \alpha} \pi_{j+1}(\beta).\label{eq:onestep}
\end{equation}
The distribution $\pi_{j}(\alpha)$ can be thought as 
the probability of sampling $\alpha \in X^{(j)}$ by first sampling some $\beta \sim
\pi_{j+1}$ and then removing one of the elements of $\beta$ uniformly at random.
The following proposition generalizes this observation, and follows from a
simple inductive argument (cf.~\cite[Proposition 2.3.1]{AlevL20})
\begin{proposition}\label{prop:multstep}
    Let $(X, \pi)$ be a simplicial complex of rank $n$. For all $0 \le j
	\le k \le n$ and $\alpha \in X^{(j)}$, one has
    \[ \pi_j(\alpha) = \frac{1}{\binom{k}{j}} \sum_{\beta
	\supset\alpha,\atop \beta \in X^{(k)}} \pi_k(\beta).\]
\end{proposition}
Similarly, given a face $\alpha \in X^{(j)}$, we define the distribution
$\pi^{(\alpha)}$ on $X_\alpha^{(n-j)}$
by conditioning $\pi$ on the containment of $\alpha$, i.e.~for
all $\alpha' \in X_\alpha^{(n- j)}$ we have,
\begin{equation}\pi_{n- j}^{(\alpha)}(\alpha') = \frac{\pi_d(\alpha \cup
	\alpha')}{\sum_{\beta \in X^{(n)},\atop \beta \supset \alpha} \pi_n(\beta)} =
    \frac{\pi_d(\alpha \cup \alpha')}{\binom{n}{j} \cdot
    \pi_j(\alpha)},\label{eq:linkk-def}
\end{equation}
where the last part is due to \cref{prop:multstep}. Analogously, we have
\begin{proposition}\label{prop:linkdef}
    Let $(X, \pi)$ be a simplicial complex of rank $n$. Let $\alpha \in
	X^{(j)}$ and $\tau \in X_\alpha^{(\ell)}$ be faces for some $0 \le \ell \le j \le
	n$.
    Then,
    \[ \pi_l^{(\alpha)}(\tau) = \frac{\pi_{j+\ell}(\alpha \cup
    \tau)}{\binom{j+\ell}{\ell} \cdot \pi_j(\alpha)}.\] 
\end{proposition}

Of particular importance to us will be the \ref{eq:link-def} graph $\Emm_\alpha \in \RR^{X^{(1)}_\alpha \times X^{(1)}_\alpha}$ given any $\alpha \in X^{(\le n - 2)}$. We recall that for all distinct pairs of vertices $x, y \in X_\alpha^{(1)}$, we have
\begin{equation} \Emm_\alpha(x, y) = \pi^{(\alpha \cup \set x)}(y) = \frac{\Pr_{\omega \sim \pi_n}\sqbr*{ \left.\omega \supset \alpha \cup \set{x, y}~\right|~\omega \supset \alpha \cup \set x}}{n - |\alpha| - 1},\tag{link}\label{eq:link-def}
\end{equation}
and $\Emm_\alpha(x, x) = 0$ for all $x \in X_\alpha^{(1)}.$

\subsection{Higher Order Random Walks on  Simplicial Complexes}\label{ss:wint}
Let $(X, \pi)$ be a simpicial complex of rank $n$. The up-down walk $\Udw_{\ell \lra n} := \UpDown_{\ell \lra n}(X, \pi)$ between the $\ell$-th and $n$-th levels, $X^{(\ell)}$ and $X^{(n)}$ respectively, is defined as the following random walk on $X^{(\ell)}$: Starting from an arbitrary face $\widehat\omega^{(0)} \in X^{(\ell)}$ for all $t \ge 1$ move from $\widehat\omega^{(t-1)}$ to $\widehat\omega^{(t)}$ according to the following simple rule,
\begin{boxedc}{0.85\textwidth}{Update Rule For the Up-Down Walk, $\Udw_{\ell \lra n}$}
\begin{itemize}
    \item sample $\omega \sim \pi_n$, conditional on $\omega \supset \widehat\omega^{(t-1)}$,
    \item draw a uniformly subset among all the subsets of $\omega$ of size $\ell$, and output it as $\widehat\omega$.
\end{itemize}
\end{boxedc}

Similarly, the down-up walk $\Duw_{n \lra \ell}$ between the $n$-th and $\ell$-th levels, $X^{(n)}$ and $X^{(\ell)}$ respectively, as the following random walk $X^{(n)}$: Starting from an arbitrary $\omega^{(0)} \in X^{(n)}$ and moves from $\omega^{(t-1)}$ to $\omega^{(t)}$ according to the following simple rule,
\begin{boxedc}{0.85\textwidth}{Update Rule for the Down-Up Walk, $\Duw_{n \lra \ell}$}
    \begin{itemize}
    \item draw a subset $\widehat\omega$ of $\omega$ of size $\ell$, uniformly at random,
    \item draw a subset $\omega \sim \pi$ conditioned on containing $\widehat\omega$, and output it as $\omega^{(t)}$.
\end{itemize}
\end{boxedc}

It is well known, \cite{AlevL20, DinurK17, DiksteinDFH18}, that the random walks $\Duw_{n \lra \ell}$ and $\Udw_{\ell \lra n}$ can be decomposed as a product of random down- and up-movements on $X$. Formally, for $0 \le \ell \le k \le n$, we define the up-walk $\Upo_{\ell \to k} := \Upp_{\ell \to k}(X, \pi)$ and the down-walk $\Doo_{k \to \ell} := \Down_{k \to \ell}(X, \pi)$ as the following random walks,
\begin{align}
    \Upo_{\ell \to k}\parens*{\widehat\omega, \omega} &~=~\pi_{k - \ell}^{(\widehat\omega)}(\omega) = \frac{\one[\omega \supset \widehat\omega] \cdot \Pr_{\widetilde\omega \sim \pi_n}\sqbr*{ \widetilde\omega \supset \omega \mid \widetilde\omega \supset \widehat\omega}}{\binom{n-\ell}{k-\ell}},\tag{up-walk}\label{eq:upw-def}\\
    \Doo_{k \to \ell}\parens*{\omega, \widehat\omega} &~=~ \frac{\one\sqbr*{\widehat\omega \subset \omega}}{\binom{k}{\ell}}.\tag{down-walk}\label{eq:downw-def}
\end{align}
\begin{fact}[Folklore]\label{fac:folklore}
    Let $(X, \pi)$ be a simplical complex of rank $n$, then writing $\Duw_{n \lra \ell} := \DownUp_{n \lra \ell}(X, \pi), \Udw_{\ell \lra n} := \Udw_{\ell \lra n}(X, \pi)$, $\Upo_{\ell \to n} = \Upp_{\ell \to n}(X, \pi)$, and $\Doo_{n \to \ell} = \Down_{n \to \ell}(X, \pi)$ for the down-up, up-down, up- and down-walks between the $n$-th and $\ell$-th levels of $X$ respectively, we have
    \begin{enumerate}
        \item $\parens*{\Upo_{\ell \to n}}^* = \Doo_{n \to \ell}$, i.e.~the operators $\Upo_{\ell \to n}$ and $\Doo_{n \to \ell}$ are adjoint operators with respect to the measures $\pi_n$ and $\pi_\ell$,
        \item $\Udw_{\ell \lra n} = \Upo_{\ell \to n} \Doo_{n \to \ell}$ --  in particular the operator $\Udw_{\ell \lra n}$ is PSD,
        \item $\Duw_{n \lra \ell} = \Doo_{n \to \ell} \Upo_{\ell \to n}$ -- in particular the operator $\Duw_{n \lra \ell}$ is PSD.
    \end{enumerate}
\end{fact}
For any $\widehat\omega \in X$ and any $0 \le\ell \le n' = n - |\widehat\omega|$, we will write $\Upo_{\widehat\omega, \ell \to n'}, \Doo_{\widehat\omega, n' \to \ell}, \Udw_{\widehat\omega, \ell \lra n'},$ and $\Duw_{\widehat\omega, n' \lra \ell}$ for the corresponding up, down, up-down, and down-up walks in the complex $(X_{\widehat\omega}, \pi^{\widehat\omega})$.
\subsection{Local to Global Analysis}
Given a simplicial complex $(X, \pi)$ of rank $n$, we define the \ref{eq:pecon} factor $\lc_\Phi(\widehat\omega)$ for any $\widehat\omega \in X^{(\le r - 2)}$ as follows,
\begin{equation}
  \lc_\Phi(\widehat\omega) := \sup\set*{ \left.\frac{\Ent^{\Phi}_{\pi^{(\widehat\omega)}_{1}}(\Upo_{\widehat\omega, 1 \to
n'}\gee)}{\Ent^{\Phi}_{\pi^{(\widehat\omega)}_{n'}}(\gee)}~\right|~ \gee \in \RR^{X_{\widehat\omega}^{(n')}}~~\textrm{and}~~n' = n - |\widehat\omega|.}  \label{eq:pecon}\tag{local $\Phi$-entropy contraction}
\end{equation}  
Equivalently, $\lc_\Phi(\widehat\omega) \in \RR_{> 0}$ is the smallest constant satisfying the equality
\[\Ent^{\Phi}_{\pi^{\widehat\omega}_{1}}(\Upo_{\widehat\omega, 1 \to n'}\gee) \le
	\lc_\Phi(\widehat\omega)
	\cdot \Ent^{\Phi}_{\pi^{(\widehat\omega)}_{
		n'}}(\gee) ~~\textrm{for all}~~\gee \in
		\RR^{X_{\widehat\omega}^{(n')}}~~\textrm{where}~~n' = n - |\widehat\omega|. \]
When $\Phi(t) = t \log t$, we will simply write $\lec(\widehat\omega)$ in place of $\lc_\Phi(\widehat\omega)$. We also make the following observation for the special case $\Phi(t) = t^2$, i.e.~when $\Ent^{\Phi}_\bullet(\bullet)$ equals the \ref{eq:v-def} functional $\Var_\bullet(\bullet)$. The following proposition is well understood,

\begin{proposition}\label{eq:varcont}
    Let $(X, \pi)$ be a simplicial complex of rank $n$. Then, for the choice of $\Phi(t) = t^2$, for any $\widehat\omega \in X^{(\le n- 2)}$ we have
    \[ \lc_\Phi(\widehat\omega) = \frac{1}{n - |\widehat\omega|} + \frac{n - |\widehat\omega| - 1}{n - |\widehat\omega|} \cdot \lambda_2(\Emm_{\widehat\omega}),\]
    where $\Emm_{\widehat\omega}$ is the \ref{eq:link-def} graph of $\widehat\omega$.
\end{proposition}
We provide a proof for \cref{eq:varcont} in \cref{app:varcont}.

A crucial tool we will be using in \cref{sec:fi} is the so called Garland method, due to \cite{Garland73}. To this end, we define the \ref{eq:localization} $\eff|_{\widehat\omega} \in \RR^{X_{\widehat\omega}(k  - j)}$ of a function $\eff \in \RR^{X^{(k)}}$ on a link $\widehat\omega \in X^{(j)}$ for $j \le k$ as the following function,
\begin{equation}\label{eq:localization}\tag{localization}
\eff|_{\widehat\omega}(\alpha) = \eff(\widehat\omega \sqcup \alpha)~~\textrm{for all}~~\alpha \in X_{\widehat\omega}^{(k - j)}.
\end{equation}

We first observe that by appealing to the chain rule for the \ref{eq:pent-defn}, one can obtain a convenient expression for it in terms of localizations.

\begin{fact}[Chain Rule for $\Phi$-Entropy]\label{fac:pent-cr}
    Let $(X, \pi)$ be a simplicial complex of rank $n$. For all $0 \le \ell \le r \le n$ and non-negative $\eff \in \RRp^{X^{(r)}}$, we have
	\[ 	
		\Ent^{\Phi}_{\pi_r}(\eff) = \Exp_{\widehat\omega \sim \pi_\ell}
		\Ent^{\Phi}_{\pi_\ell}(\eff|_{\widehat\omega})
		+ \Ent^{\Phi}_{\widehat\omega \sim \pi_\ell}\parens*{ \Exp_{\alpha \sim \pi^{(\widehat\omega)}_{r
	-\ell}} \eff|_{\widehat\omega}(\alpha)}.\]
    In particular,
    \[\Ent^{\Phi}_{\pi_r}(\eff) = \Exp_{\widehat\omega \sim \pi_\ell}
		\Ent^{\Phi}_{\pi_\ell}(\eff|_{\widehat\omega}) + \Ent_{\pi_\ell}\parens*{ \Upo_{\ell \to r} \eff},\]
  where $\Upo_{\ell \to r} := \Upp_{\ell \to r}(X, \pi)$ is the \ref{eq:upw-def} on $X$
\end{fact}
\begin{proof}
	The result follows by straight-forward computation,
	\begin{align*}
		\Ent^{\Phi}_{\pi_r}(\eff)
		&~=~\Exp_{\widehat\omega \sim \pi_\ell} \Exp_{\alpha \sim \pi^{(\widehat\omega)}_{r - \ell}}
		\Phi(\eff|_{\widehat\omega}(\alpha))
		- \Phi\parens*{ \Exp_{\widehat\omega \sim \pi_\ell} \Exp_{\alpha \sim \pi^{(\widehat\omega)}_{r- \ell}}
		\eff|_{\widehat\omega}(\alpha)},\\
		&~=~\Exp_{\widehat\omega \sim \pi_\ell} \Ent^{\Phi}_{\pi^{(\widehat\omega)}_{r-\ell}}(\eff|_{\widehat\omega})
		+ \Exp_{\widehat\omega \sim \pi_\ell} \Phi\parens*{ \Exp_{\alpha \sim \pi^{(\widehat\omega)}_{r - \ell}}
			\eff|_{\widehat\omega}(\alpha)} - \Phi\parens*{ \Exp_{\widehat\omega\sim \pi_\ell}  \Exp_{\alpha\sim
		\pi^{(\widehat\omega)}_{r - \ell}} \eff|_{\widehat\omega}(\alpha)},\\
		&~=~\Exp_{\widehat\omega \sim \pi_\ell} \Ent^{\Phi}_{\pi_\ell}(\eff|_{\widehat\omega})
		+ \Ent^{\Phi}_{\widehat\omega \sim \pi_\ell}\parens*{ \Exp_{\alpha \sim \pi^{(\widehat\omega)}_{r
		-\ell}} \eff|_{\widehat\omega}(\alpha)}.
	\end{align*}
    The second statement follows from the definition of the \ref{eq:upw-def} $\Upo_{\ell \to r}$.
\end{proof}

We also recall the following identities,
\begin{lemma}\label{lem:garland}
    Let $(X, \pi)$ be a simplicial complex of rank $n$. Writing $\Duw_{n \lra r} = \DownUp_{n \lra r}(X, \pi)$, $\Udw_{n-1} = \UpDown_{n-1 \lra n}(X, \pi)$, and $\Emm_{\widehat\omega}$ for the \ref{eq:link-def} of the face $\widehat\omega \in X^{(\le n-2)}$, for all $\eff \in \RR^{X^{(n)}}$ and $\ell \le r \le n$, we have
    \begin{enumerate}
        \item $\inpr*{\eff,  \eff}_{\pi_n} = \Exp_{\widehat\omega \sim \pi_\ell} \inpr*{\eff|_{\widehat\omega}, \eff|_{\widehat\omega}}_{\pi_{n- \ell}^{(\widehat\omega)}}$,
        \item $\inpr*{\eff, \Duw_{n \lra r} \eff}_{\pi_n} = \Exp_{\widehat\omega \sim \pi_\ell} \inpr*{ \eff|_{\widehat\omega}, \Duw_{\widehat\omega,n - \ell \lra r - \ell} \eff|_{\widehat\omega}}_{\pi_{n - \ell}^{(\widehat\omega)}}$,
        \item $\inpr*{\eff, \Udw_{n-1} \eff}_{\pi_n} = \Exp_{\widehat\omega \sim \pi_{n-2}}\parens*{ \inpr*{ \eff|_{\widehat\omega}, \parens*{ \frac{\Ide}{n} + \frac{n-1}{n} \cdot \Emm_{\widehat\omega}} \eff|_{\widehat\omega}}_{\pi_1^{(\widehat\omega)}}}$
    \end{enumerate}
\end{lemma}
Items (1) and (3) are folklore results, we refer the reader to \cite[Lemma 3.7]{AlevL20} for a proof. The second item is a straight-forward consequence of the definition of the \ref{eq:upw-def} $\Upo_{\bullet_2 \to \bullet_1}$ and that $\Duw_{\bullet_1 \lra \bullet_2} = \Doo_{\bullet_1 \to \bullet_2} \Upo_{\bullet_2 \to \bullet_1}$. In particular, for any $\widehat\omega \in X^{(\ell)}$ and any $\omega \supset \widehat\omega$ we have by the definition of the \ref{eq:upw-def},
\begin{equation} [\Upo_{\widehat\omega, r - \ell\to n-\ell} \eff|_{\widehat\omega}](\omega \setminus \widehat\omega) = [\Upo_{r \to n}\eff](\omega).
 \end{equation}

We recall the following result due to \cite{Lee23},
\begin{lemma}[Theorem 3.5, \cite{Lee23}]\label{lem:lees}
    Let $(X, \pi)$ be an $n$-partite simplicial complex and let $0 \le k < n$.  Suppose we have $\EC(\Upo_{n - 1 \to n}) \ge (Cn)^{-1}$ for some $C \in \RRpp$. Then, we have
    \[ \EC\parens*{\Upo_{k \to k+1}} \ge \frac{1}{(k+1)(C + 1)},\]
    where we have written $\Upo_{\ell \to k} = \Upp_{\ell \to k}\parens*{ X, \pi}$ for the \ref{eq:upw-def} between $X^{(\ell)}$ and $X^{(k)}$
\end{lemma}
 \section{Expanderized Random Walks}\label{sec:expander}
Let $(X, \pi)$ be an $n$-partite simplicial complex. For any $\ell \le n$, the up-down walk $\Udw_{\ell \lra n} := \UpDown_{\ell \lra n}(X, \pi)$ on the $\ell$-th level $X^{(\ell)}$ of $X$ introduced in \cref{ss:wint} admits the following alternative description: Starting from an arbitrary face $\widehat\omega^{(0)} \in X^{(\ell)}$ move from $\widehat\omega^{(t-1)}$ to $\widehat\omega^{(t)}$ according to the following simple rule,
\begin{boxedc}{0.85\textwidth}{Update Rule for the Up-Down Walk, $\Udw_{\ell \lra n}$}
\begin{itemize}
    \item sample $\omega \sim \pi$, conditional on $\omega \supset \widehat\omega^{(t-1)}$,
    \item sample $S \sim \uni_{\binom{[n]}{\ell}}$,
    \item output $\widehat\omega^{(t)} = \omega_{S}$.
\end{itemize}
\end{boxedc}
We will expanderize the up-down walk $\Udw_{\ell \lra n}$ in the following manner: Given a $k$-regular labelled graph $H$ on the vertex set $\binom{[n]}{\ell}$, we will denote the $a$-th neighbor of vertex $v$ by $\Out_H(v, a)$. We define the expanderized up-down walk $\Papx_{\ell \lra n} = \UpDown_{\ell \lra n}(X, \pi, H)$ as the walk which starts from an arbitrary face $\widehat\omega^{(0)}$ and moves from $\widehat\omega^{(t-1)}$ to $\widehat\omega^{(t)}$ according to the following simple rule,
\begin{boxedc}{0.85\textwidth}{Update Rule for the Expanderized Up-Down Walk, $\Papx_{\ell \lra n}$}
\begin{itemize}
    \item sample $\omega \sim \pi$, conditional on $\omega \supset \widehat\omega^{(t-1)}$,
    \item sample $a \sim \uni_{[k]}$ and set $S= \Out_H(\typ(\widehat\omega^{(t-1)}), a)$,\footnote{Where we recall that the type of $\widehat\omega$ is the sides of the simplicial complex that $\widehat\omega$ intersects}
    \item output $\widehat\omega^{(t)} = \omega_S$.
\end{itemize}
\end{boxedc}

Similarly, the down-up walk $\Duw_{n \lra \ell}$ between the $n$-th level $X^{(n)}$ and the $\ell$-th level $X^{(\ell)}$ of an $n$-partite simplicial complex $(X, \pi)$ introduced in \cref{ss:wint} admits the following alternative description: Start from $\omega^{(0)} \in X^{(n)}$ and move from $\omega^{(t-1)}$ to $\omega^{(t)}$ according to the following simple rule,
\begin{boxedc}{0.85\textwidth}{Update Rule for the Down-Up Walk, $\Duw_{n \lra \ell}$}
\begin{itemize}
    \item sample $S \sim \uni_{\binom{[n]}{\ell}}$ uniformly at random,
    \item set $\widehat\omega = \omega_{S}$,
    \item set $\omega^{(t)}$ to be a random face drawn from $\pi$, conditional on containing $\widehat\omega$.
\end{itemize}
\end{boxedc}
Similarly, we define the expanderized down-up walk $\Paqx_{n \lra \ell} = \DownUp_{n \lra \ell}(X, \pi, H)$ to be the random walk on $X^{(n)} \times \binom{[n]}{\ell}$, starting from an arbitrary face-subset pair $(\omega^{(0)}, S^{(0)})$ and move from $(\omega^{(t-1)}, S^{(t-1)})$ to $(\omega^{(t)}, S^{(t)})$ according to the following simple rule,
\begin{boxedc}{0.85\textwidth}{Update Rule for the Expanderized Down-Up Walk, $\Paqx_{n \lra \ell}$}
\begin{itemize}
    \item sample $a \sim \uni_{[k]}$ and set $S' = \Out_H(S^{(t-1)}, a)$,
    \item set $\widehat\omega = \omega_{S'}$,
    \item set $\omega^{(t)} \sim \pi$ to be a random face conditional on containing $\widehat\omega$,
    \item sample $b \sim \uni_{[k]}$ and set $S^{(t)} = \Out_H(S^{(t-1)}, b)$,
    \item output $(\omega^{(t)}, S^{(t)})$.
\end{itemize}
\end{boxedc}
For convenience we also define the expanderized down- and up-walks given a degree regular labelled graph $H = (\binom{[n]}{\ell}, E)$ 
\[\Qdo_{n \to \ell} = \Down_{n \to \ell}(X, \pi, H) \in \RR^{(X^{(n)} \times \binom{[n]}{\ell}) \times X^{(\ell)}}~~\textrm{and}~~\Qup_{\ell \to n} = \Upp_{\ell \to [n]}(X, \pi, H) \in \RR^{X^{(\ell)} \times (X^{(n)} \times \binom{[n]}{\ell})},\]
as follows,
\[ \Qdo_{n \to \ell}( (\omega, S), \widehat\omega ) = \frac{\one[ S \sim_H \typ(\widehat\omega) ]}{k} \cdot \one[\omega \supset \widehat\omega]~~\textrm{for all}~~\omega \in X^{(n)}, \widehat\omega \in X^{(\ell)}, S \in \binom{[n]}{\ell},\]
and $\Qup_{\ell \to n} = \parens*{ \Qdo_{n \to \ell} }^*$ where the adjoint is taken with respect to the distributions $\pi_n \tensor \uni_{\binom{[n]}{\ell}}$ and $\pi_{\ell}$, i.e.~
\[ \Qup_{\ell \to n}\parens*{ \widehat\omega, (\omega, S) } = \frac{\one[ S \sim_H \typ(\widehat\omega) ] }{k } \cdot \Pr_{\widetilde\omega \sim \pi_n}\sqbr*{ \widetilde\omega = \omega \mid \omega \supset \widehat\omega }~~\textrm{for all}~~\omega \in X^{(n)}, \widehat\omega \in X^{(\ell)}, S \in \binom{[n]}{\ell},\]
and the notation $T \sim_H S$ is used to denote the adjacency relation in the graph $H$, i.e.~$\set*{S, T} \in E(H)$. 

We summarize the random movements described by the expanderized up- and down-walks in words as follows: The expanderized down-walk $\Qdo_{n \to \ell}$ first samples a random neighbor of $T$ of $S$ in $\binom{[n]}{\ell}$, and then restricts the coordinates of $\omega$ to $T$, i.e.~moves to $\omega_T$. The expanderized up-walk $\Qup_{\ell \to n}$ on the other hand first samples a facet $\omega \in X^{(n)}$ from $\pi$ conditional on containing $\widehat\omega$ and after picking a random neighbor $S$ of $\typ(\widehat\omega)$ in $H$ moves to $(\omega, S)$.
\begin{proposition}
    For any $n$-partite pure simplicial complex $(X, \pi)$ and a $k$-regular labelled graph $H = (\binom{[n]}{\ell}, E)$, writing $\Qdo_{n \to \ell} = \Down_{n \to \ell}(X, \pi, H)$ and $\Qup_{\ell \to n} = \Upp_{\ell \to n}(X, \pi, H)$ we have,
    \[ \parens*{\pi_n \tensor \uni_{\binom{[n]}{\ell}}} \Qdo_{n \to \ell} = \pi_{\ell}~~\textrm{and}~~\pi_{\ell} \Qup_{\ell \to n} = \pi_n \tensor \uni_{\binom{[n]}{\ell}}.\]
\end{proposition}

\begin{proof}
    Let $(\omega, S) \sim \pi_n \tensor \uni_{\binom{[n]}{\ell}}$ be a random sample. Notice that a random neighbor of $S$ in $H$ is still distributed uniformly at random as the uniform distribution stationary for the random walk over a $k$-regular graph, q.v.~\cref{fac:baby}. Thus, conditional on $\omega$, a single step of $\Qdo_n$ ends up restricting $\omega$ to a random set of coordinates $S$ -- this precisely yields the distribution $\pi_{\ell}$, q.v.~\cref{prop:multstep}. 

    The second statement follows since $\Qup_{\ell \to n}$ is the adjoint operator, q.v.~\cref{prop:adjoint-defn}.
\end{proof}

The following is easy to verify,
\begin{corollary}\label{cor:e2v}    For any $n$-partite simplicial complex $(X, \pi)$ and $k$-regular labelled graph $H = \parens*{ \binom{[n]}{\ell}, E)}$,
    \begin{itemize}
    \item $\UpDown_{\ell \lra n}(X, \pi, H^2) = \Upp_{\ell \to n}(X, \pi, H) \cdot \Down_{n\to \ell}(X, \pi, H)$,
    \item $\DownUp_{n \lra \ell}(X, \pi, H) = \Down_{n \to \ell}(X, \pi, H) \cdot \Upp_{\ell \to n}(X, \pi, H)$.
    \end{itemize}
\end{corollary}

We now summarize several useful properties of the expanderized up- and down-walks,

\begin{corollary}\label{cor:stat}
Let $(X, \pi)$ be an $n$-partite complex and $H = (\binom{[n]}{\ell}, E)$ a $k$-regular graph. For any $\ell \le n$, writing $\Papx_{\ell \lra n} = \UpDown_{\ell \lra n}(X, \pi, H^2)$, $\Paqx_{n \lra \ell} = \DownUp_{n \lra \ell}(X, \pi, H)$ $\Qdo_{n \to \ell} = \Down_{n \to \ell}(X, \pi, H)$ and $\Qup_{\ell \to n} = \Upp_{\ell \to n}(X, \pi, H)$ we have,
\begin{enumerate}
    \item $(\pi_n \tensor \uni_{\binom{[n]}{\ell}}) \Paqx_{n \lra \ell} = \pi_n \tensor \uni_{\binom{[n]}{\ell}}$, i.e.~$\pi_n \tensor \uni_{\binom{[n]}{\ell}}$ is the stationary distribution of $\Paqx_{n \lra \ell}$,
    \item $\pi_{\ell} \Papx_{\ell \lra n} = \pi_{\ell}$, i.e.~$\pi_{\ell}$ is the stationary distribution of $\Papx_{\ell \lra n}$.
    \item $\Paqx_{n \lra \ell}$ and $\Papx_{\ell \lra n}$ are PSD operators.
    \item $\Papx_{\ell \lra n}$ and $\Paqx_{n \lra \ell}$ are self-adjoint operators.
\end{enumerate}
\end{corollary}
Since in our proofs it will be more convenient to use $\Papx_{\ell \lra n} := \UpDown_{\ell \lra n}(X, \pi, H)$ directly, initialized with $H$ and not $H^2$, we also note the following.
\begin{proposition}\label{prop:stat}
Let $(X, \pi)$ be an $n$-partite complex and $H = (\binom{[n]}{\ell}, E)$ a $k$-regular graph. Then, the expanderized up-down walk $\Papx_{\ell \lra n} := \UpDown_{\ell \lra n}(X, \pi, H)$ has the stationary distribution $\pi_\ell$ and is reversible.
\end{proposition}
As \cref{cor:stat} reaches the same end by replacing $H$ with $H^2$ we will postpone the proof of \cref{prop:stat} to \cref{ap:stat}.

Now, we present the results we prove for expanderized random walks. Our first result shows that the expanderized up-down walk approximates the usual up-down walk in the \ref{eq:opn-def},
\begin{theorem}\label{thm:exp-close}
Let $(X, \pi)$ be an $n$-partite simplicial complex and let $H$ be a $k$-regular labelled graph on the vertex set $\binom{[n]}{\ell}$. Writing $\Papx_{\ell \lra n} := \UpDown(X, \pi, H)$ and $\Udw_{\ell \lra n} := \UpDown(X, \pi)$ for the expanderized- and the regular up-down walks on $X^{(\ell)}$, we have
    \[ \norm*{ \Papx_{\ell \lra n} - \parens*{1 - \lambda(H)} \cdot \Udw_{\ell \lra n} }_{\opp, \pi_{\ell}} \le \lambda(H). \]
\end{theorem}
We present the proof of \cref{thm:exp-close} in \cref{ss:pexp}.

\cref{thm:exp-close} immediately implies the following bounds for the \ref{eq:gap-def} of expanderized walks,
\begin{corollary}\label{cor:gaplift}
Let $(X, \pi)$ be an $n$-partite simplicial complex and let $H$ be a $k$-regular labelled graph on the vertex set $\binom{[n]}{\ell}$. Writing $\Papx_{\ell \lra n} := \UpDown(X, \pi, H)$, $\Udw_{\ell \lra n} = \UpDown_{\ell \lra n}(X, \pi)$, $\Paqx_{n \lra \ell} = \DownUp(X, \pi, H)$, and $\Duw_{n \lra \ell} = \DownUp_{n \lra \ell}(X, \pi)$, we have
    \begin{align*}
    \Gap\parens*{ \Papx_{\ell \lra n} } &~\ge~\Gap\parens*{ \Udw_{\ell \lra n} } \cdot \Gap^\star(H),\\
    \Gap\parens*{ \Paqx_{n \lra \ell}} &~\ge~\Gap\parens*{ \Duw_{n \lra \ell} } \cdot \Gap^\star(H^2),
    \end{align*}
    where $\Gap^\star(G) = 1 - \lambda(G)$ and $\lambda(G)$ denotes the \ref{eq:two-sided} of the graph $G$.
\end{corollary}
Unfortunately, a bound on the \ref{eq:gap-def} is in many settings not enough to obtain optimal mixing time bounds. We show however, that \cref{thm:exp-close} allows us to transfer log-Sobolev inequalities (\ref{eq:lsi-def}) for the usual up-down walks to the expanderized up-down walks,
\begin{corollary}\label{cor:entropicstuff}
    Let $(X, \pi)$ be an $n$-partite simplicial complex and let $H$ be a $k$-regular labelled graph on the vertex set $\binom{[n]}{\ell}$. Writing $\Papx_{\ell \lra n} := \UpDown_{\ell \lra n}(X, \pi, H)$ and $\Udw_{\ell \lra n} := \UpDown_{\ell \lra n}(X, \pi)$ for the up-down walk on $X^{(\ell)}$, we have
    \[ \LS\parens*{ \Papx_{\ell \lra n} } \ge \LS\parens*{ \Udw_{\ell \lra n} } \cdot \Gap^\star(H),\]
    where $\Gap^\star(H) = 1 - \lambda(H)$ and $\lambda(H)$ denotes the \ref{eq:two-sided} of the graph $H$.
\end{corollary}
We will prove \cref{cor:entropicstuff} in \cref{ss:lsiexp}. We state a convenient corollary of \cref{cor:entropicstuff} which immediately follows from \cref{lem:lsfun}, \cref{cor:stat}, and the data processing inequality \cref{fac:dpi},
\begin{corollary}\label{cor:baked}
     Let $(X, \pi)$ be an $n$-partite simplicial complex and let $H$ be a $k$-regular labelled graph on the vertex set $\binom{[n]}{\ell}$. Then, writing $\Qdo_{n \lra \ell} = \Down_{n \to \ell}(X, \pi, H)$ and $\Udw_{\ell \lra n} = \UpDown_{\ell \lra n}(X, \pi)$ we have,
     \[ \EC(\Qdo_{n \lra \ell}) \ge \LS(\Udw_{\ell \lra n}) \cdot \Gap^\star(H^2).\]

     In particular, we have for $\Papx_{\ell \lra n} := \UpDown_{\ell \lra n}(X, \pi, H^2)$ and $\Paqx_{n \lra \ell} = \DownUp_{n \lra \ell}(X, \pi, H)$,
     \[ \EC(\Papx_{\ell \lra n}) \ge \LS(\Udw_{\ell \lra n}) \cdot \Gap^\star(H^2)~~\textrm{and}~~\EC(\Paqx_{n \lra \ell}) \ge \LS(\Udw_{\ell \lra n}) \cdot \Gap^\star(H^2),\]
     where $\Gap^\star(H^2) = 1 - \lambda(H^2)$ and $\lambda(H^2)$ denotes the \ref{eq:two-sided} of the graph $H^2$.
\end{corollary}
As we will see in \cref{sec:hijack}, \cref{cor:baked} will indeed allow us to prove optimal mixing time bounds for the expanderized walks in many cases of interest.
\subsection{Expanderized Up-Down Walk Approximates the Up-Down Walk Well, Proof of \cref{thm:exp-close}}\label{ss:pexp}
\begin{proof}[Proof of \cref{thm:exp-close} ]
For convenience, we will write $\Papx := \Papx_{\ell \lra n}$ and $\Udw := \Udw_{\ell \lra n}$.    Let $\Emm$ denote the random-walk matrix of the graph $H$ where each transition occurs with the probability $1/k$ and $\Jay$ for the random-walk matrix of the clique over $\binom{[n]}{\ell}$ with self-loops, i.e.~$\Jay = \one\one^\top/\binom{[n]}{\ell}$. We will write, $\lambda := \lambda(\Emm).$ 

Let $S \in \binom{[n]}\ell$ be arbitrary and suppose some $\bar\omega \in X^{(\ell)}$ is given such that $\typ(\bar\omega) = S$. 

For all $\eff \in \RR^{X^{(\ell)}}$, we have 
\[ [\Papx \eff](\bar\omega) = \sum_{\widehat\omega \in X[S^c]}
\Pr_{\omega \sim \pi}[\omega_{S^c} = \widehat\omega \mid \omega_S = \bar\omega] \cdot
\sum_{a \in [k]} \frac{\eff\parens*{ (\bar\omega \oplus \widehat\omega)_{\Out_H(S, a)}}}{k}. \]
Similarly, we have
\[ [\Udw \eff](\bar\omega) =
\sum_{\widehat\omega \in X[S^c]}
\Pr_{\omega \sim \pi}[\omega_{S^c} = \widehat\omega \mid \omega_S = \bar\omega]\sum_{T \in 
\binom{[n]}{\ell}} \frac{\eff\parens*{ (\bar\omega \oplus \widehat\omega)_{T}}}{\binom{n}{\ell}}.\]
For any given facet $\omega \in X^{(n)}$ we define the function $\gee_\omega \in \RR^{\binom{[n]}{\ell}}$ as, $\gee_\omega(T) = \eff\parens*{ \omega_T }$.

We have,
\[ [\Emm\gee_\omega](T) = \sum_{a \in [k]} \frac{\eff\parens*{ \omega_{\Out_H(T, a)}}}{k}~~\textrm{and}~~[\Jay \gee_\omega](i) = \sum_{T \in \binom{[n]}{\ell}} \frac{\eff\parens*{\omega_{T}}}{\binom{n}{\ell}}. \]
Thus, we have
\begin{align}
    [\Papx \eff](\bar\omega) &~=~\sum_{\widehat\omega \in X[S^c]} \Pr_{\omega \sim \pi}[\omega_{S^c} = \widehat\omega \mid \omega_{S} = \bar\omega] \cdot [\Emm\gee_{\bar\omega \oplus \widehat\omega}](S),\label{eq:der-garl}\\
    [\Udw \eff](\bar\omega) &~=~\sum_{\widehat\omega \in X[S^c]} \Pr_{\omega \sim \pi}[\omega_{S^c} = \widehat\omega \mid \omega_{S} = \bar\omega]\cdot [\Jay\gee_{\bar\omega \oplus \widehat\omega}](S).\label{eq:full-garl}
\end{align}
In particular combining \cref{eq:der-garl} and \cref{eq:full-garl} and noticing that for $\omega \sim \pi$ the law of $\omega_{S^c}$ conditional on $\omega_S = \bar\omega$ is given by $\pi_{n - \ell}^{(\widehat\omega)}$,
\begin{align*} \norm*{ (\Papx - (1-\lambda)\Udw)\eff}^2_{\pi_{\ell}} 
&~=~ \Exp_{\bar\omega \sim \pi_{\ell}} \parens*{\Exp_{\widehat\omega \sim \pi^{(\bar\omega)}_{n - 
\ell}}\sqbr*{ [\Emm\gee_{\bar\omega \oplus \widehat\omega}](\typ(\bar\omega)) - (1- \lambda) [\Jay\gee_{\bar\omega \oplus \widehat\omega}](\typ(\bar\omega))}}^2,\\
&~\le~ \Exp_{\bar\omega \sim \pi_{\ell}} \Exp_{\widehat\omega \sim \pi^{(\bar\omega)}_{n - \ell}}\sqbr*{ \parens*{ [\Emm\gee_{\bar\omega \oplus \widehat\omega}](\typ(\bar\omega)) - (1- \lambda) [\Jay\gee_{\bar\omega \oplus \widehat\omega}](\typ(\bar\omega))}^2}.
\end{align*}
where the last inequality is obtained by appealing to Jensen's inequality and the convexity of $t \to t^2$.

Now, we observe the law of $\bar\omega \oplus \widehat\omega$ obtained by first sampling $\bar\omega \sim \pi_\ell$ and then $\widehat\omega \sim \pi^{(\bar\omega)}$ is given by $\pi$.  Furthermore, any $\omega \in X^{(n)}$ occurs exactly $\binom{[n]}{\ell}$ times in the expectation above -- once for each $S \in \binom{[n]}{\ell}$ acting as $\typ(\bar\omega)$, which happens with probability $\binom{n}{\ell}^{-1}$. Thus,
\begin{align*}
    \norm*{ (\Papx - (1-\lambda)\Udw)\eff}^2_{\pi_{\ell}} 
    &~\le~\Exp_{\omega \sim \pi_n}\sqbr*{ \Exp_{S \sim \uni_{\binom{[n]}{\ell}}} \parens*{[(\Emm - (1- \lambda)\Jay)\gee_\omega](S)}^2},\\
    &~\le~\Exp_{\omega\sim \pi_n}\sqbr*{\lambda^2 \cdot \norm{\gee_\omega}^2_{\uni_{\binom{[n]}{\ell}}}}
\end{align*}
where the last inequality is due to $\norm{\Emm - (1- \lambda) \Jay}_{\opp, \uni_{\binom{[n]}{\ell}}} \le \lambda$ as $\lambda(\Emm) \le \lambda$!

Now, we finally note
\[ \Exp_{\omega \sim \pi}\sqbr*{ \norm{\gee_\omega}^2_{\uni_{\binom{[n]}{\ell}}}} = \Exp_{\omega \sim \pi}\sqbr*{ \frac{1}{\binom{n}{\ell}} \sum_{S \in \binom{[n]}{\ell}} \eff(\omega_{S})^2 } = \Exp_{\bar\omega \sim \pi_{\ell}} \eff(\bar\omega)^2 = \norm{\eff}^2_{\pi_{\ell}},\]
The last equality is due to the observation that first sampling $\omega \sim \pi$ and then outputting $\omega_{S}$ for $S \sim \uni_{\binom{[n]}{\ell}}$ picked uniformly at random amounts to simply sampling $\bar\omega \sim \pi_{\ell}$, q.v.~\cref{prop:multstep}.

In particular,
\[ \norm*{ (\Papx - (1- \lambda) \Udw)\eff }_{\pi_{\ell}} \le \lambda \cdot \norm{\eff}_{\pi_{\ell}}.\]
As $\eff$ was picked arbitrarily, this allows us to conclude the proof of our theorem by appealing to the definition of the \ref{eq:opn-def}.
\end{proof}

\subsection{A Spectral Gap Bound For Expanderized Higher Order Random Walks, Proof of \cref{cor:gaplift}}

\begin{proof}[Proof of \cref{cor:gaplift}]
   Suppose $\eff \in \RR^{X^{(\ell)}}$ is a unit vector, i.e.~$\norm*{\eff}_{\pi_\ell} = 1$,  satisfying $\Papx_{\ell \lra n
   } \eff = \lambda(\Papx_{\ell \lra n}) \eff.$ In particular, $\inpr*{\eff, \one}_{\pi_\ell} = 1$.  By \cref{prop:stat}, the stationary distribution of $\Papx$ and $\Udw_{\ell \lra n}$ are the same -- thus we must have $\norm{\Udw_{\ell \lra n}\eff}_{\pi_\ell} \le \lambda(\Udw_{\ell \lra n}).$

We have by the triangle inequality and \cref{thm:exp-close},
\begin{align*} \lambda(\Papx_{\ell \lra n}) &~=~\norm*{\Papx_{\ell \lra n} \eff}_{\pi_n},\\
&~\le~ \norm{(\Papx_{\ell \lra n} - (1-  \lambda(H))\Udw_{\ell \lra n})\eff}_{\pi_{\ell}} + \norm*{ (1 - \lambda(H)) \Udw_{\ell \lra n} \eff}_{\pi_{\ell}},\\
&~=~\lambda(H) + (1- \lambda(H)) \cdot \lambda(\Udw_{\ell \lra n}). 
\end{align*}

Thus, we have
\[ \Gap(\Papx_{\ell \lra n}) \ge 1 - \lambda(H) - (1 -\lambda(H)) \cdot \lambda(\Udw_{\ell \lra n}) = (1-  \lambda(H)) \cdot (1 - \lambda(\Udw_{\ell \lra n})) = \Gap^*(H) \cdot \Gap(\Udw_{\ell \lra n}),\]
where we have used that $\Udw_{\ell \lra n}$ is PSD to obtain the last equality, q.v.~\cref{fac:folklore}. The analogous statement for $\Paqx_{n \lra \ell}$ follows from \cref{fac:switcheroo} and \cref{cor:e2v}.
\end{proof}

\subsection{Log-Sobolev and Entropy Contraction Bounds For Expanderized Walks, Proof of \cref{cor:entropicstuff}}\label{ss:lsiexp}

\begin{proof}[Proof of \cref{cor:entropicstuff}]
    Let $\eff \in \RRp^{X^{(\ell)}}$ be an arbitrary function satisfying $\Ent_{\pi_{\ell}}(\eff^2) \ne 0$. We have,
    \begin{align*}
        \frac{\inpr*{\eff, \parens*{\Ide - \Papx_{\ell \lra n}} \eff}_{\pi_{\ell}}}{\Ent_{\pi_{\ell}}(\eff)^2}
        &~=~\frac{\inpr*{\eff, \parens*{ \Ide - (1-  \lambda(H)) \cdot \Udw_{\ell \lra n}}\eff}_{\pi_{\ell}}}{\Ent_{\pi_{\ell}}(\eff^2)}
        + \frac{\inpr*{ \eff,\parens*{(1 - \lambda(H)) \Udw_{\ell \lra n} - \Papx_{\ell \lra n}} \eff}_{\pi_{\ell}}}{\Ent_{\pi_{\ell}}(\eff^2)}.
    \end{align*}
    Notice that by \cref{thm:exp-close}, we should have
    \[ \inpr*{ \eff,\parens*{(1 - \lambda(H)) \Udw_{\ell \lra n} - \Papx_{\ell \lra n}} \eff}_{\pi_{\ell}} \ge - \lambda(H) \cdot \inpr*{ \eff , \Ide \eff}_{\pi_{\ell}}. \]
    Thus,
    \begin{align*}
    \frac{\inpr*{\eff, \parens*{\Ide - \Papx_{\ell \lra n}} \eff}_{\pi_{\ell}}}{\Ent_{\pi_{\ell}}(\eff)^2}
        &~\ge~ \frac{ \inpr*{ \eff, \parens*{\Ide - \parens*{ (1 - \lambda(H)) \cdot \Udw_{\ell \lra n} + \lambda(H) \cdot \Ide }} \eff}_{\pi_{\ell}}}{\Ent_{\pi_{\ell}}(\eff^2)},\\
        &~\ge~\LS\parens*{ (1 - \lambda(H)) \cdot \Udw_{\ell \lra n} + \lambda(H) \cdot \Ide},\\
        &~=~\LS(\Udw_{\ell \lra n}) \cdot \Gap^\star(H),
    \end{align*}
    where the last inequality is by noticing,
    \[ \inpr*{ \eff, \parens*{\Ide - \parens*{ a \cdot \Ide + (1 - a) \Pii}} \eff }_\mu = (1- a) \cdot \inpr*{\eff, \parens*{\Ide - \Pii} \eff}_\mu,\]
    and the definition of the log-Sobolev inequality (\ref{eq:lsi-def}) and the log-Sobolev constant (\cref{eq:lsc-def}) . Appealing to the definition of the log-Sobolev inequality (\ref{eq:lsi-def}) once again yields the result.
\end{proof}

\section{Functional Inequalities on Simplical Complexes}\label{sec:fi}
In this section, we will prove several functional inequalities involving the down-up walk $\Duw_{n \lra \ell}$. For convenience we define the set $\Cc_\ell(X)$ as the set of $\ell$-chains in $X$, i.e.~the collection of  sequences
\[ \varnothing := \omega^{(0)} \subsetneq \omega^{(1)} \subsetneq \cdots \subsetneq \omega^{(\ell)} \in X^{(\ell)}, \]
such that $\omega^{(i)} \in X^{(i)}$ for all $i = 0, \ldots, \ell$. 
Similarly, for $x \in X^{(1)}$ we define $\Cc_\ell(x)$ as the set of $\ell$-chains in $X$ starting from $x \in X^{(1)}$, i.e.~the collection of sequnces
\[ x =: \omega^{(1)} \subsetneq \omega^{(2)} \subsetneq \cdots \subsetneq \cdots \subsetneq \omega^{(\ell)},\]
such that $\omega^{(i)} \in X^{(i)}$ for all $i = 0, \ldots, \ell$.

\begin{theorem}\label{thm:pent-c}
    For any $\widehat\omega \in X$, we set $n_{\widehat\omega} = n - |\widehat\omega|$
    and set $\pi_{\widehat\omega, \star} = \min_{\widetilde\omega \in X_{\widehat\omega}^{n_z}} \pi_{\widehat\omega}^{(n_{\widehat\omega})}(\widetilde\omega)$.
	We have for all convex $\Phi: \RRp \to \RRp$ and $\eff \in \RR^{X^{(n)}}$.
	{\small
	\begin{equation}\Ent^{\Phi}_{\pi_\ell}\parens*{ \Upo_{\ell \to n} \eff }
		\le \parens*{ 1 - \min\set*{\left. \prod_{j = 0}^{\ell - 1} (1 -
					\lc_\Phi(\omega^{(j)}))
				~\right|~  \varnothing =: \omega^{(0)} \subsetneq \omega^{(1)} \subsetneq \cdots \subsetneq \omega^{(\ell)} \in
	\Cc_{\ell}(X)}} \cdot
	\Ent_{\pi_r}^{\Phi}(\eff).\label{eq:pentl}
	\end{equation}
}
	Equivalently,
	\begin{equation}\CF_\Phi( \Upo_{\ell \to n} ) \ge 
		\min\set*{\left. \prod_{j = 0}^{\ell - 1} (1 - \lc_\Phi(\omega^{(j)}))
				\right|  \varnothing =: \omega^{(0)} \subsetneq \omega^{(1)} \subsetneq \cdots \subsetneq \omega^{(\ell)} \in
	\Cc_{\ell}(X) }.\label{eq:pentf}
\end{equation}
    In particular, writing $\lc^{(i)}_\Phi(X, \pi) = \max_{\widehat\omega \in X^{(i)}} \lc_\Phi(\widehat\omega)$, we have
    \[ \CF_\Phi\parens*{ \Upo_{\ell \to n} } \ge \prod_{j = 0}^{\ell - 1} \parens*{1 - \lc^{(j)}_\Phi(X, \pi)}.\]
\end{theorem}
As mentioned before our proof is inspired by the exposition in \cite{ChenE22} and follows the Garland method, \cite{Garland73}.  After submitting our results to arxiv, it came to our attention that the same proof technique for proving \cref{thm:pent-c} already appeared in \cite{Liu23} in the context of variance contraction -- though only for the case $\ell = n-1$. We will list a few immediate consequences of \cref{thm:pent-c}. The following bound is immediate given \cref{eq:varcont} and \cref{thm:pent-c},
\begin{corollary}[Spectral Gap Bound]\label{thm:mixing}
    Let $(X, \pi)$ be a simplicial complex of rank $n$. We have,
    \begin{equation}\label{eq:mixx}
        \Gap(\Duw_{n \leftrightarrow \ell}) \ge \frac{n - \ell}{n} \cdot \min \set*{\left. \prod_{i =
    0}^{\ell-1}\Gap(\Emm_{z_i})~\right|~ \varnothing =: \omega^{(0)} \subsetneq \omega^{(1)} \subsetneq \cdots \subsetneq \omega^{(\ell)} \in
        \Cc_{\ell - 1}(X) }.
    \end{equation}
	In particular, writing $\Gap_k(X, \pi) := \min_{x \in X^{(k)}} \Gap(\Emm_x)$
    we have
    \[ \Gap(\Duw_{n \leftrightarrow \ell}) \ge \frac{n - \ell}{n} \cdot \prod_{i = 0}^{\ell - 1} \Gap_i(X, \pi).\]
\end{corollary}

We also prove a useful lemma that shows we can directly relate the entropy contraction constant to the log-Sobolev constant of the down-up walk,
\begin{lemma}\label{lem:lsi}
    Let $(X, \pi)$ be a simplicial complex of rank $n$. For any $\widehat\omega \in X$, we set 
    \[ \pi^{\star}_{\widehat\omega, k} = \min_{\widetilde\omega \in X_{\widehat\omega}^{(k)}} \pi^{(\widehat\omega)}_{k}(\widetilde\omega),~~ \Gap_{n-2}(X, \pi) = \min_{\widehat\omega \in X^{(n-2)}} \Gap(\Emm_{\widehat\omega}),~~\textrm{and}~~C_{\widehat\omega, k} = \begin{cases}
    1 & \pi^\star_{\widehat\omega, k} > 1/2, \\
    \frac{1 - 2\pi^\star_{\widehat\omega, k}}{\log\parens*{\parens*{\pi^\star_{\widehat\omega,k}}^{-1}-1}} &\textrm{otherwise}.
\end{cases}\]
    where $\Emm_{\widehat\omega}$ is the \ref{eq:link-def} of $\widehat\omega$ and $\Gap(\bullet)$ denotes the \ref{eq:gap-def}.
    
	Recalling that $\EC(\bullet)$ denotes the \ref{eq:pent-con} for $\Phi(t) = t \log t$, we have
    \begin{align*}
		\LS(\Duw_{n \lra \ell})&~\ge~
		\min\set*{\left. C_{\omega^{(\ell)}, n - \ell} 
  ~\right|~\omega^{(\ell)} \in X^{(\ell)}} \cdot \EC\parens*{\Upo_{\ell \to n}},\\
    \LS(\Udw_{n-1}) &~\ge~\frac{n-1}{n}\cdot \min\set*{\left. C_{\omega^{(n-2)}, 1} 
  ~\right|~\omega^{(n-2)} \in X^{(n-2)}} \cdot \Gap_{n-2}(X, \pi) \cdot \EC(\Upo_{n-2 \to n-1}).
    \end{align*}
	In particular, writing $\lec_i(X, \pi) := \min_{\widehat\omega \in X^{(i)}} \lec(\widehat\omega)$
	and $C_{\ell, k} = \min_{\widehat\omega \in X^{(\ell)}} C_{\widehat\omega, k}$,
		\begin{align*}
			\LS(\Duw_{n \lra \ell})&~\ge~C_{\ell, n - \ell} \cdot 
		\prod_{i = 0}^{\ell - 1} (1 - \lec_i(X, \pi))
		\end{align*}
\end{lemma}

We will prove this result in \cref{ss:lsi}.

We notice that \cref{thm:pent-c} in conjunction with \cref{lem:mlsfun} immediately implies the following corollary,
\begin{corollary}
	Let $(X, \pi)$ be an $n$-partite simplicial complex, 
	\[ \MLS(\Duw_{n \lra \ell}) \ge \EC(\Upo_{\ell \to n}) \ge \prod_{j = 0}^{\ell - 1} (1 -
	\lec_i(X)).\]
\end{corollary}

\subsection{Proof of $\Phi$-Entropy Contraction Bounds, \cref{thm:pent-c}}
\begin{proof}[Proof of \cref{thm:pent-c}]
	For $\ell = 0$, the LHS is equal to 0 (\cref{rem:aoe}), thus we see the
	product in \cref{eq:pentl} and \cref{eq:pentf} is taken over an empty set
	and equals 1. Thus, equality holds in this case with $\CF_\Phi(\Upo_{0 \to
	n}) = 1$. We proceed by induction on the rank of the simplicial complex.
	We have by the chain rule for \ref{eq:pent-defn} (\cref{fac:pent-cr}),
		\[ \Ent^{\Phi}_{\pi_\ell}(\Upo_{\ell \to n}\eff) = 
			\Exp_{x \sim \pi_1} \Ent^{\Phi}_{\pi^{(x)}_{\ell - 1}}( \Upo_{x, \ell-1
		\to n -1} \eff|_x) + \Ent^{\Phi}_{\pi_1}(\Upo_{1 \to n} \eff).\]
		Let $c := \min_{x \sim X^{(1)}} \CF_\Phi(\Upo_{x, \ell - 1 \to n-
		1})$. By the induction hypothesis,
		\begin{equation}
			c\ge \min\set*{ \left. \prod_{j = 1}^{\ell - 1} (1 - \lc_\Phi(\omega^{(j)})) 
					\right| x \in X^{(1)}, x =:\omega^{(1)} \subsetneq \omega^{(2)} \subsetneq \cdots \subsetneq \omega^{(\ell-1)}\in
				\Cc_{\ell-1}(x)}.\label{eq:pent-ih}
		\end{equation}

	Hence,	we obtain,
	\begin{align*}
		\Ent^{\Phi}_{\pi_\ell}(\Upo_{\ell \to n} \eff)
		&~\le~ (1-c)\Exp_{x \sim \pi_1} \Ent^{\Phi}_{\pi^{(x)}_{n-1}}(
		\eff|_x )
		+  \Ent^{\Phi}_{\pi_1}(\Upo_{1 \to n} \eff).
	\end{align*}
	Now, using the chain-rule (\cref{fac:pent-cr}) for \ref{eq:pent-defn} once more, we
	have $\Exp_{x \sim \pi_1} \Ent^{\Phi}_{\pi^{(x)}_{n -1}}(
	\eff|_x ) = \Ent^{\Phi}_{\pi_n}(\eff) - \Ent_{\pi_1}(\Upo_{1 \to n} \eff)$.
	Substituting this in,
	\begin{align*}\Ent^{\Phi}_{\pi_\ell}(\Upo_{\ell \to n} \eff) 
		&~=~
		(1-c)\cdot \parens*{ \Ent_{\pi_n}^{\Phi}(\eff) -
			\Ent^{\Phi}_{\pi_1}(\Upo_{1 \to
		n} \eff)} + \Ent^{\Phi}_{\pi_1}(\Upo_{1 \to n} \eff),\\
		&~=~(1-c)\cdot\Ent^{\Phi}_{\pi_n}(\eff) + c \cdot \Ent_{
\pi_1}^{\Phi}(\Upo_{1 \to n}\eff).
\end{align*}
Now, using $\Ent_{\pi_1}^{\Phi}(\Upo_{1 \to n}\eff) \le \lc_\Phi(\varnothing)\cdot
\Ent_{\pi_n}^{\Phi}(\eff)$ we obtain
\[ \Ent^{\Phi}_{\pi_\ell}(\Upo_{\ell \to n} \eff) \le \parens*{1 - c \cdot
(1-\lc_\Phi(\varnothing))} \cdot \Ent^{\Phi}_{\pi_n}(\eff).\]
	Now, the statement follows from \cref{eq:pent-ih}.
\end{proof}

\subsection{Proof of the log-Sobolev Inequality, \cref{lem:lsi}}\label{ss:lsi}

\begin{proof}[Proof of \cref{lem:lsi}]
We follow a similar strategy to what we have followed to establish \cref{thm:pent-c}. We have,
\begin{align*}
    \inpr*{\eff, \parens*{\Ide - \Duw_{n \lra \ell}} \eff}_{\pi_n}
    &~=~\inpr*{\eff, \eff}_{\pi_n} - \inpr*{ \Upo_{\ell \to n} \eff, \Upo_{\ell \to n} \eff}_{\pi_\ell},\\
    &~=~\Exp_{\widehat\omega \sim \pi_{\ell}}\sqbr*{ \inpr*{\eff|_{\widehat\omega}, \eff|_{\widehat\omega}}_{\pi^{(\widehat\omega)}_{n - \ell}} - \Exp_{\widehat\omega \sim \pi_\ell}\inpr*{ \Upo_{\widehat\omega, 0 \to n - \ell}\eff|_{\widehat\omega}, \Upo_{\widehat\omega, 0 \to n - \ell}\eff|_{\widehat\omega}}_{\pi^{(\widehat\omega)}_{n -\ell}}},\\
    &~=~\Exp_{\widehat\omega \sim \pi_\ell}\sqbr*{ \inpr*{ \eff|_{\widehat\omega}, \parens*{ \Ide - \Duw_{\widehat\omega, n - \ell \lra 0}} \eff|_{\widehat\omega} }_{\pi_{n - \ell}^{(\widehat\omega)}}},\\
    &~=~\Exp_{\widehat\omega \sim \pi_{\ell}} \sqbr*{\inpr*{ \eff|_{\widehat\omega}, \parens*{\Ide - \Jay_{\pi^{(\widehat\omega)}_{n - \ell}}}\eff|_{\widehat\omega}}_{\pi^{(\widehat\omega)}_{n-\ell}}},
\end{align*}
where we have used Items (1) and (2) of \cref{lem:garland} to obtain the second equality and have written $\Jay_\mu = \one\mu$ for the clique with respect to $\mu$.

    Now, by \cref{fac:evb}, we have
    $\LS\parens*{\Jay_{\pi^{(\widehat\omega)}_{n-\ell}}} \ge C_{\widehat\omega, n - \ell}$ --
    where $C_{\widehat\omega, n - \ell}$ is defined as in the statement of \cref{lem:lsi}. Thus, writing $C_{\ell, n - \ell} := \min_{\widehat\omega \in X^{(\ell)}} C_{\widehat\omega, n - \ell}$, we have
    \begin{align*}
        \inpr*{\eff,\parens*{ \Ide - \Duw_{n \lra \ell}} \eff}_{\pi_n} 
        &~\ge~C_{\ell, n - \ell} \cdot \Exp_{\widehat\omega \sim \pi_\ell} \Ent_{\pi^{(\widehat\omega)}_{n - \ell}}\parens*{ \eff^2|_{\widehat\omega} },\\
        &~\ge~C_{\ell, n - \ell} \cdot \parens*{ \Ent_{\pi_n}(\eff^2) - \Ent_{\pi_\ell}\parens*{\Upo_{\ell \to n} \eff^2}},
    \end{align*}
    where we have used the chain rule for entropy, \cref{fac:pent-cr}, to obtain the last statement.

    Now, using the definition of \ref{eq:pent-con}, i.e.~writing for $\Phi(t) = t \cdot \log t$,
    \[ \Ent_{\pi_\ell}\parens*{ \Upo_{\ell \to n}\eff^2} \le \parens*{1 - \EC\parens*{\Upo_{\ell \to n}}} \cdot \Ent_{\pi_n}\parens*{\eff^2}.\]
    Thus,
    \[ \inpr*{\eff, \parens*{\Ide - \Duw_{n \lra \ell}} \eff}_{\pi_n} \ge C_{\ell, n - \ell} \cdot \EC\parens*{\Upo_{\ell \to n}} \cdot \Ent_{\pi_n}\parens*{ \eff^2}.\]
    Now, the first statement follows by appealing to the definition of the log-Sobolev inequality (\ref{eq:lsi-def}) and the log-Sobolev constant (\cref{eq:lsc-def}). The second statement concerning $\Duw_{ n \lra \ell}$ now immediately follows from \cref{thm:pent-c}.

    To obtain the log-Sobolev inequality for $\Udw_{n-1}$, we make use of Items (1) and (3) in \cref{lem:garland} and proceed as above. We have,
    \begin{align*}
        \inpr*{ \eff, \parens*{\Ide - \Udw_{n-1}} \eff}_{\pi_{n-1}} &~=~\frac{n-1}{n} \cdot\Exp_{\widehat\omega \sim \pi_{n-2}}\sqbr*{ \inpr*{ \eff|_{\widehat\omega}, \parens*{ \Ide - \Emm_{\widehat\omega}} \eff|_{\widehat\omega} }_{\pi_{1}^{(\widehat\omega)}}}
    \end{align*}
    Now, appealing to \cref{fac:evb}, we obtain $\LS(\Emm_{\widehat\omega}) \ge \Gap(\Emm_{\widehat\omega}) \cdot C_{\widehat\omega, 1}$ for all $\widehat\omega \in X^{(n-2)}$. Thus,
    \begin{align} \inpr*{ \eff, \parens*{\Ide - \Udw_{n-1}} \eff}_{\pi_{n-1}} &~\ge~\frac{n-1}{n} \cdot C_{n-2, 1} \cdot \Gap_{n-2}(X, \pi) \cdot \Exp_{\widehat\omega \sim \pi_{n-2}} \Ent_{\pi^{(\widehat\omega)}_1}(\eff^2|_{\widehat\omega}),\notag\\
    &~=~\frac{n-1}{n} \cdot C_{n-2, 1} \cdot \Gap_{n-2}(X, \pi) \cdot \parens*{ \Ent_{\pi_{n-1}}(\eff^2) - \Ent_{\pi_\ell}\parens*{ \Upo_{n-2 \to n-1} \eff^2}},\notag\\
    &~=~\frac{n-1}{n} \cdot C_{n-2, 1} \cdot \Gap_{n-2}(X, \pi) \cdot \EC\parens*{\Upo_{n-2 \to n -1}} \cdot \Ent_{\pi_{n-1}}(\eff^2),\label{eq:yestag}
    \end{align}
    where we have appealed to the chain rule for entropy, \cref{fac:pent-cr}, to obtain the first equality. 
\end{proof}

\section{Application: Sampling Using the Expanderized Walks}\label{sec:hijack}
In the present section, we prove that the expanderized walks rapidly mix for the (i) list-coloring problem and (ii) Ising models with bounded interaction matrix. First, we describe the random sampling problems we are interested in mention the state of the art sampling results we are interested in expanderizing, and state our results. We will then presents proofs for our applications in \cref{ss:bdedg} and \cref{ss:ising} respectively.

A list coloring instance $(G, \Lap)$ consists of a graph $G = (V, E)$ and a collection of colours $\Lap = (L(v))_v$ for every vertex. A valid list coloring of $(G, \Lap)$ is then a set of pairs $\set*{(v, c(v))}_{v \in V}$ satisfying the following two conditions,
\begin{enumerate}
    \item $c(v) \in L(v)$ for all vertices $v \in L$,
    \item $c(u) \ne c(v)$ for all edges $\set*{u, v} \in E$.
\end{enumerate}
We will write $(X^{(G, \Lap)}, \uni^{(G, \Lap)})$ for the simplicial complex of proper list coloring of $(G, \Lap)$ weighted by the uniform distirbution $(G, \Lap)$ on all list colorings, i.e.~
\[X^{(G, \Lap)} = \set*{\left. \alpha \subset \bigsqcup_{v \in V} \set*{v} \times L(v)~\right|~\textrm{there exists a proper list coloring $\chi$ of $(G, \Lap)$ such that}~\alpha \subset \chi}.\]
We will show that the expanderized walks rapidly mix when sampling list colorings of bounded degree graphs. Further, the lower bound in the number of colors matches with the state of the art -- see \cite{ChenDMPP19, Liu21, BlancaCPCPS22}.
\begin{theorem}\label{thm:colhijack}
Let $(G, \Lap)$ be a list-coloring instance where $G = (V, E)$ is a graph on $n$ vertices of maximum degree $\Delta \le O(1)$ and $H_n$ be a labelled graph on $[n]$ of constant \ref{eq:two-sided} $\lambda(H_n)$ bounded away from 1. Then, for some absolute constant $\ee \approx 10^{-5}$,\footnote{See \cite{ChenDMPP19}} and any $K = O(1)$, if $(11/6 + K)\Delta \ge |L(v)| \ge (11/6 - \ee) \cdot \Delta$ for all vertices $v \in V$, the mixing time of the expanderized walks $\Papx_{n- 1} = \UpDown_{\ell \lra n}(X^{(G, \Lap)}, \uni^{(G, \Lap)}, H^2)$ and $\Paqx_n = \DownUp_{n \lra \ell}(X^{(G, \Lap)}, \uni^{(G, \Lap)}, H)$ satisfies,
\[ \Tmix(\Papx_{n-1}, \ee) \le C_1 \cdot n \parens*{\log n + \log \ee^{-1} }~~\textrm{and}~~\Tmix(\Paqx_n, \ee) \le C_2 \cdot n \parens*{ \log n + \log \log \ee^{-1}},\]
where $C_1$ and $C_2$ are universal constants not depending on $n$ but on $\Delta$.
\end{theorem}
\begin{remark}
    By \cref{thm:alon}, we can pick a constant degree graph as the graph $H_n$ in the statement of \cref{thm:colhijack}. Thus, a single step of the random walk can be implemented using $O(1)$-random bits -- making the total number of random bits used in the random walk $O(n \log n)$. In contrast, the standard down-up walk or the up-down walk requires $O(\log n)$ random bits to perform a single step, and $O(n \log^2n)$ random bits in total.
\end{remark}

We recall that the Ising model $\mu_{\Jay, \etch}: \set*{+1, -1}^n \to \RRp$ with interaction matrix $\Jay \in \RR^{n \times n}$ and external field $\etch \in \RR^n$ from statistical physics is a probability distribution on the hypercube satisfying,
\begin{equation}
    \mu_{\Jay, \etch}(\eks) = \frac{ \exp\parens*{ \frac{1}{2} \inpr*{\eks, \Jay \eks}_{\ell_2} + \inpr*{\etch, \eks}_{\ell_2}}}{Z(\Jay, \etch)}~~\textrm{where}~~Z(\Jay, \etch) = \sum_{\eks \in \set*{+1, -1}^n} \exp\parens*{ \frac{1}{2} \inpr*{\eks, \Jay \eks}_{\ell_2} + \inpr*{\etch, \eks}_{\ell_2}}
\end{equation}
We notice that we can identify any $\eks \in \set*{+1, -1}^n$ with a value by using the encoding,
\[ \eks^{\pm} = \set*{ (i, \eks(i)) \mid i \in [n] }. \]
Thus, we define the simplicial complex $(X^{(\Jay, \etch)}, \mu_{\Jay, \etch})$, where
\[ X^{(\Jay, \etch)} = \set*{ \alpha \subset [n] \setminus \set*{\pm 1} \mid \textrm{ for each $i \in [n]$, $\alpha$ contains at most one element $(i, x)$}}.\]

We show that our expanderize walks mix rapidly assuming that the external field $\etch \in \RR^n$ is well-behaved, i.e.~$\norm{\etch}_{\ell_\infty}$ does not grow with $n$,
\begin{theorem}\label{thm:ishijack}
Let $(X^{(\Jay, \etch)}, \mu_{\Jay, \etch})$ be the simplicial complex defined above corresponding to the Ising model defined by the interaction matrix $\Jay \in \RR^{n \times n}$ and external field $\etch \in \RR^n$ and $H_n$ a constant degree graph whose \ref{eq:two-sided} is a constant bounded away from 1. Under the assumption that $\Jay$ is PSD and satisfies $\norm*{\Jay}_\opp \le 1$, the following hold,  
{\small
\[ \Tmix(\Papx_{n-1}, \ee) \le \frac{ O\parens*{\norm*{\etch}_{\ell_\infty}} \cdot n}{(1 - \norm*{\Jay}_\opp)^2} \parens*{\log( n  + \norm*{\etch}_{\ell_1} )+ \log \ee^{-1} }~~\textrm{and}~~\Tmix(\Paqx_n, \ee) \le \frac{O\parens*{\norm*{\etch}_{\ell_\infty}} \cdot n}{(1 - \norm*{\Jay}_\opp)^2} \parens*{\log( n  + \norm*{\etch}_{\ell_1} )+ \log \ee^{-1} },\]
}
where the $O(\bullet)$ notation hides a universal constant not depending on $n, \Jay,$ or $\etch$. Furthermore, the term $(1 - \norm*{\Jay}_\opp)^2$ in the denominator can be replaced with $(1 -\norm*{\Jay}_\opp)(1- \theta)$ if the maximum operator norm of any two-by two principal submatrix of $\Jay$ is $\theta$.
\end{theorem}
\begin{remark}
    By \cref{thm:alon}, we can pick a constant degree graph as the graph $H_n$ in the statement of \cref{thm:ishijack}. Thus, ignoring numerical difficulties in simulating biased coins, a single step of the random walk can be implemented using $O(1)$-random bits -- making the total number of random bits used in the random walk $O(n \log n)$ when $\norm*{\etch}_{\ell_\infty} = O(1)$. In contrast, the standard down-up walk or the up-down walk requires $O(\log n)$ random bits to perform a single step and $O(n \log^2n)$ random bits in total.
\end{remark}
\subsection{List Coloring of Bounded Degree Graphs}\label{ss:bdedg}

We make the following observations about the complex associated to proper list colorings,
\begin{proposition}[Folklore]\label{prop:colfolk}
    Let $(G = (V, E), \Lap)$ be a list-coloring instance. Let $K \in \NN$ satisfy $\deg(v) + K_+ \ge |L(v)| \ge \deg(v) + K_-$ for all $v \in V$. Then, writing $(Y, \pi) := (X^{(G, \Lap)}, \uni^{(G, \Lap)})$ we have,
    \[ \lambda_2\parens*{ \Emm_{\widehat\chi} } \le \frac{1}{K_-}~~\textrm{for all}~~\widehat\chi \in Y^{(n-2)},\]
    where $\Emm_{\widehat\chi}$ is the \ref{eq:link-def} of the face $\widehat\chi$.
    
    Similarly, for any $\widehat\chi \in Y^{(n-2)}$, we have $\min_{(u, c) \in Y_{\widehat\chi}^{(1)}} \pi_1^{(\widehat\chi)}(u, c) \ge \frac{K_-}{(\Delta + K_+)^2}$ where $\Delta = \max_{v \in V} \deg(v)$.
\end{proposition}

\begin{proof}[Proof Sketch]
    The face $\widehat\chi$ fixes the color of all but two distinct vertices $u, v \in V$. Writing $A_{\widehat\chi}(u)$ and $A_{\widehat\chi}(v)$ for the colors available to $u$ and $v$ by assigning every other vertex a color according to $\widehat\chi$, we observe that the graph $G_{\widehat\chi}$ which underlies the \ref{eq:link-def} $\Emm_{\widehat\chi}$ is bipartite with the partition $\set u \times A_{\widehat\chi}(u)$ and $\set v \times A_{\widehat\chi}(v)$. Further, if $\set{u, v} \not\in E$, this graph is a complete bipartite graph and thus $\lambda_2(\Emm_{\widehat\chi}) = 0$. Otherwise, the only edges missing from this graph are the pairs,
    \begin{equation}\label{eq:missing}\set*{\set*{ (u, c), (v, c)} \mid c \in A_{\widehat\chi}(u) \cap A_{\widehat\chi}(v)}.
    \end{equation}
    In particular, the edges missing from $G_{\widehat\chi}$ form a matching. Thus, the adjacency matrix of $G_{\widehat\chi}$ is of the form $\Bee - \Aye$ where $\Bee$ is the adjacency matrix of the complete bipartite graph on this bipartition and $\Bee$ and $\Aye$ corresponds to the adjacency matrix corresponding to the edges in $\Aye$. Writing $\Dee$ for the degree matrix of $G_{\widehat\chi}$, we observe $\Emm_{\widehat\chi} = \Dee^{-1}(\Bee - \Aye)$.

    Now, by using eigenvalue-interlacing can conclude
    \[ \lambda_2(\Emm_{\widehat\chi})  = \lambda_2(\Dee^{-1} \cdot (\Bee - \Aye)) \le \parens*{\lambda_2(\Bee) + \norm*{\Aye}_{\opp}} \cdot \norm*{\Dee^{-1}}_\opp\]
    We notice now that $\lambda_2(\Bee) = 0$ as it is the adjacency matrix of the complete bipartite graph, $\norm{\Aye}_\opp = 1$ as it is the adjacency matrix of a matching, and $\norm*{\Dee^{-1}}_\opp = 1/d$ where $d$ is the minimum number of colors $c' \in A_{\widehat\chi}(v) \setminus \set c$ for any $c \in A_{\widehat\chi}(u)$, or vice versa. We notice that a color $c \in L(v)$ is precisely not in $A_{\widehat\chi}(v)$, because there is a neighbor $w \ne v$ of $u$ colored with the same color in $\widehat\chi$. There are $\deg(u) - 1$ such neighbors, thus $d \ge |L(u)| - \deg(u) \ge K_-$, which concludes the proof of the claim about the eigenvalue.

    Let $c \in A_{\widehat\chi}(u)$ be arbitrary. Suppose $\set{u,v} \not\in E$. Then, it is easy to observe
    \[ \pi_1^{(\widehat\chi)}(u, c) = \frac{1}{|A_{\widehat\chi}(u)|} \ge \frac{1}{K_-},\]
    as in a random coloring conditional on $\widehat\chi$ any color $c \in A_{\widehat\chi}(u)$ is equally likely.
    
    If however $\set{u,v} \in E$, then we can easily see that there are $|A_{\widehat\chi}(u)| \cdot |A_{\widehat\chi}(v)| - |A_{\widehat\chi}(u) \cap A_{\widehat\chi}(v)|$ ways of completing $\widehat\chi$ to a full coloring, and $|A_{\widehat\chi}(v) \setminus \set c|$ of them have $u$ colored with $c$. Thus,
    \[ \pi^{(\widehat\chi)}_1(u, c) \ge \frac{ |A_{\widehat\chi}(v) \setminus \set c|}{|A_{\widehat\chi}(u)| \cdot |A_{\widehat\chi}(v)| - |A_{\widehat\chi}(u) \cap A_{\widehat\chi}(v)|} \ge \frac{K_-}{(\max\set*{\deg(u), \deg(v)}+ K_+)^2}.\]
\end{proof}
We recall the following result of \cite{Liu21, BlancaCPCPS22},
\begin{theorem}[Theorem 1.2, \cite{Liu21, BlancaCPCPS22}]
    Let $(G, \Lap)$ be a list-coloring instance where $G = (V, E)$ is a graph on $n$ vertices of maximum degree $\Delta \le O(1)$. Then, for some absolute constant $\ee \approx 10^{-5}$,\footnote{See \cite{ChenDMPP19}} if $|L(v)| \ge (11/6 - \ee) \cdot \Delta$ for all vertices $v \in V$, then the \ref{eq:gap-def}, modified log-Sobolev (\cref{eq:mlsc-def}), and the log-Sobolev constants (\cref{eq:lsc-def}) of the down-up walk $\Duw_n = \DownUp_{n \lra n -1}(X^{(G, \Lap)}, \uni^{(G, \Lap)})$ on the collection of proper list colorings is all $\Omega(n^{-1})$.
\end{theorem}
Then, the following corollary immediately follows by \cref{lem:lsfun} and \cref{fac:folklore},
\begin{corollary}\label{cor:colent}
    Let $(G, \Lap)$ be a list-coloring instance where $G = (V, E)$ is a graph on $n$ vertices of maximum degree $\Delta \le O(1)$. Then, for some absolute constant $\ee \approx 10^{-5}$, if $|L(v)| \ge (11/6 - \ee) \cdot \Delta$ for all vertices $v \in V$, then the up-operator $\Upo_{n-1 \to n} = \Upp_{n-1 \to n}\parens*{X^{(G, \Lap)}, \uni^{(G, \Lap)}}$ on the collection of proper list colorings of $(G, \Lap)$ satisfies $ \EC(\Upo_{n-1 \to n}) \ge \Omega(n^{-1})$.
\end{corollary}
\begin{proof}[Proof of \cref{thm:colhijack}] Notice that by \cref{lem:lees}, \cref{cor:colent} implies that $\EC( \Upo_{n- 2 \to n- 1}) \ge \Omega(n^{-1})$ since by \cref{prop:colfolk} when $\Delta = O(1)$, $C_{n-2} = \Omega(1)$, we have $\Gap_{n-2}(X^{(G, \Lap)}, \uni^{(G, \Lap)}) = \Omega(1)$, by invoking \cref{lem:lsi} we obtain that the up-down walk $\Udw_{n-1} = \UpDown_{n-1 \lra n}(X^{(G, \Lap)}, \uni^{(G, \Lap)})$ satisfies, $\LS(\Udw_{n-1}) \ge \Omega(n^{-1})$. Then, by \cref{cor:baked} and the assumption that the \ref{eq:two-sided} $\lambda(H_n)$ is a constant bounded away from 1, we obtain $\EC(\Papx_{n}), \EC(\Paqx_{n-1}) \ge \Omega(n^{-1})$. The result, concerning mixing times follows using \cref{thm:entmix} and the observation that the state space for both walks is of size at most $n \cdot \parens*{(K + 11/6) \cdot \Delta}^n$.
\end{proof}

\subsection{The Ising Model with Bounded Correlations}\label{ss:ising}
We recall that the Ising model $\mu_{\Jay, \etch}: \set*{+1, -1}^n \to \RRp$ from statistical physics is a probability distribution on the hypercube satisfying,
\begin{equation}
    \mu_{\Jay, \etch}(\eks) = \frac{ \exp\parens*{ \frac{1}{2} \inpr*{\eks, \Jay \eks}_{\ell_2} + \inpr*{\etch, \eks}_{\ell_2}}}{Z(\Jay, \etch)}~~\textrm{where}~~Z(\Jay, \etch) = \sum_{\eks \in \set*{+1, -1}^n} \exp\parens*{ \frac{1}{2} \inpr*{\eks, \Jay \eks}_{\ell_2} + \inpr*{\etch, \eks}_{\ell_2}}
\end{equation}

Quite recently, it was shown that the down-up walk $\Duw_n = \DownUp_n(X^{(\Jay, \etch)}, \etch)$ rapidly mixes whenever $\Jay$ is a PSD matrix of small enough \ref{eq:opn-def}, i.e.~$\norm{\Jay}_\opp \ll 1$. 

\begin{theorem}[\cite{EldanKZ22, AnariJKP22, Lee23}]\label{thm:is-mix}
    Let $(X^{(\Jay, \etch)}, \mu_{\Jay, \etch})$ be the simplicial complex corresponding to the Ising model defined by the interaction matrix$\Jay \in \RR^{n \times n}$ and external field $\etch \in \RR^n$. Under the assumption that $\Jay$ is PSD and satisfies $\norm*{\Jay}_\opp \le 1$, the following hold,
    \[ \Gap(\Duw_n) \ge \frac{1-  \norm*{\Jay}_\opp}{n}~~\textrm{and}~~\EC\parens*{ \Upo_{n-1 \to n} } \ge \frac{1 - \norm*{\Jay}_\opp}{n},\]
    where $\Duw_n = \DownUp_{n \lra n-1}\parens*{ X^{(\Jay, \etch)}, \mu_{\Jay, \etch}}$ and $\Upo_{n-1 \to n} = \Upp_{n - 1 \to n}\parens*{X^{(\Jay, \etch)}, \mu_{\Jay, \etch}}$.
\end{theorem}
The \ref{eq:gap-def} bound above is due to \cite[Theorem 1]{EldanKZ22} and implies a mixing time bound of $O\parens*{\frac{n}{1 - \norm*{\Jay}_\opp} (n + \norm*{\etch}_{\ell_1})}$. This mixing time bound was subsequently improved to $O\parens*{ \frac{n\log n}{1-  \norm{\Jay}_\opp}}$ by \cite{AnariJKP22} through a modified log-Sobolev inequality and a clever argument utilizing the \emph{approximate exchange property} -- which intially appeared in \cite{AnariLOVV21} -- allowing them to bypass the dependence on $\etch$ completely. The concrete statement about entropy contraction was shown in \cite[Theorem 4.1]{Lee23}.

Now, we make the following observations,
\begin{proposition}
    \label{prop:is-conds}
    Let $(Y, \mu) := (X^{(\Jay, \etch)}, \mu_{\Jay, \etch})$ be the simplicial complex defined above corresponding to the Ising model defined by $\Jay \in \RR^{n \times n}$ and $\etch \in \RR^n$. Under the assumption that $\Jay$ is PSD and satisfies $\norm*{\Jay}_\opp \le 1$, the following hold,
    \begin{enumerate}
        \item We have $\Gap_{n-2}\parens*{ Y, \mu } = \min_{\widehat\omega \in Y^{(n-2)}} \Gap\parens*{ \Emm_{\widehat\omega}} \ge 1- \theta$ where $\Emm_{\widehat\omega}$ is the \ref{eq:link-def} of $\widehat\omega$ in $(Y, \mu)$ and $\theta$ is the maximum \ref{eq:opn-def} of any principal minor of $\Jay$. Notice that $\theta \le \norm*{\Jay}_\opp$.
        \item We have that for any $\widehat\omega \in Y^{(n-2)}$, $\min_{x \in Y_{\widehat\omega}^{(1)}}\mu^{(\widehat\omega)}_1(x) \ge \frac{1}{2} \cdot e^{-4 \cdot \norm*{\etch}_{\ell_\infty} - 1}.$
    \end{enumerate}
\end{proposition}
\begin{proof}
    For convenience we write $(Y, \mu) = (X^{(\Jay, \etch)}, \mu_{\Jay, \etch})$. Let $\widehat\omega \in Y^{(n-2)}$ be arbitrary and suppose $\typ(\widehat\omega) = [n] \setminus \set*{a, b}$. Then, we observe that we still have an ising model at our hands, for $\Jay_{\widehat\omega} \in \RR^{2 \times 2}$ and $\etch' \in \RR^{2 \times 2}$,  where 
    \[ \Jay_{\widehat\omega} = \begin{pmatrix} J(a, a) & J(a, b)\\ J(b, a) & J(b, b) \end{pmatrix}~~\textrm{and}~~\etch_{\widehat\omega}(x) = \begin{cases} \etch(a) & \textrm{ if } x = a\\ \etch(b) & \textrm{ if } x = b. \end{cases}\]
    In particular, by \cref{thm:is-mix}, we have that the down-up walk $\Duw_{\widehat\omega, 2 \lra 1}$ satisfies \ref{eq:gap-def} greater than $\frac{1 - \norm*{\Jay_{\widehat\omega}}_\opp}{2}$. We recall that this down-up walk can be described as $\Doo_{\widehat\omega, 2 \to 1} \Upo_{\widehat\omega \to 1}.$ Thus, by \cref{fac:switcheroo}, we have
    \[ \frac{1 - \norm*{\Jay_{\widehat\omega}}_\opp}{2} \le \Gap\parens*{ \Duw_{\widehat\omega, 2 \lra 1}} = \Gap\parens*{ \Upo_{\widehat\omega, 1 \to 2} \Doo_{\widehat\omega, 2 \to 1}} = \Gap\parens*{ \frac{\Ide}{2} + \frac{\Emm_{\widehat\omega}}{2}} = \frac{\Gap\parens*{ \Emm_{\widehat\omega} }}{2}.\]
    In particular, by assumption we have $\Gap(\Emm_{\widehat\omega}) \ge 1 - \norm*{\Jay_{\widehat\omega}} \ge 1 - \theta$. Notice that by eigenvalue interlacing, we always have $\norm*{\Jay_{\widehat\omega}}_{\opp} \le \norm*{\Jay}_\opp$ which establishes the bound on $\Gap(\Emm_{\widehat\omega}) \ge 1 - \norm*{\Jay}_\opp$ in the worst case.

    For the second statement, we note that for all $\eks \in \set*{-1, 1}^{2}$ we have
    \[ e^{-2 \norm*{\etch}_{\ell_\infty}}\le \exp\parens*{ \frac{1}{2} \cdot \inpr*{ \eks, \Jay_{\widehat\omega} \eks}_{\ell_2} + \inpr*{\etch_{\widehat\omega}, \eks}_{\ell_2}} \le e^{2 \norm*{\etch}_{\ell_\infty} + 1}, \]
    where we have used,
    \begin{itemize}
        \item $0 \le \inpr*{\eks, \Jay \eks}_{\ell_2} \le 1$ since $\norm*{\eks}_{\ell_2}^2 = 2$, $\norm*{\Jay_{\widehat\omega}}_\opp \le 1$, and $\Jay_{\widehat\omega}$ is PSD.
        \item $-2 \le \inpr*{ \etch, \eks }_{\ell_2} \le 2$ since $\norm*{\eks}_{\ell_1} = 2$.
    \end{itemize}
    Notice now, assuming for example that $a$ is represented by the first variable in $\eks$, writing $\Xi(\eks) = \exp\parens*{ 0.5 \cdot \inpr*{ \eks, \Jay_{\widehat\omega} \eks}_{\ell_2} + \inpr*{ \etch_{\widehat\omega}, \eks}_{\ell_2}}$
    \begin{align*} \mu_1^{(\widehat\omega)}(a, 1)&~ =~\frac{ \Xi(+1, +1) + \Xi(+1, -1)}{Z(\Jay_{\widehat\omega}, \etch_{\widehat\omega}) = \Xi(+1, +1) + \Xi(+1, -1) + \Xi(-1, +1) + \Xi(-1, -1)},\\
    &~\ge~\frac{2 \cdot e^{-2 \norm*{\etch}_\infty }}{4  \cdot e^{2 \norm*{\etch}_\infty + 1}},\\
    &~\ge~\frac{1}{2} \cdot e^{-4 \cdot \norm*{\etch}_{\ell_\infty} - 1}.
    \end{align*}
    as an analogous argument follows for all $(a, \pm 1)$ and $(b, \pm 1)$ the argument follows.
\end{proof}

\begin{proof}[Proof of \cref{thm:ishijack}]
    By invoking \cref{thm:is-mix} and \cref{lem:lees} with $C = (1 - \norm*{\Jay}_\opp)^{-1}$, we note that the \ref{eq:upw-def} $\Up_{n-2 \to n -1} = \Upp_{n-2 \to n- 1}(Y, \mu)$ satisfies,
    \[ \EC\parens*{ \Upo_{n - 2 \to n- 1} } \ge \frac{1}{n-1} \cdot \frac{1 - \norm*{\Jay}_\opp}{2 - \norm*{\Jay}_\opp} \ge \frac{1-\norm*{\Jay}_\opp}{2 (n-1) }.\]
    Now, by appealing to \cref{lem:lsi} and \cref{prop:is-conds}, we get that the log-Sobolev constant (\cref{eq:lsc-def}) of the up-down walk $\Udw_{n-1} = \UpDown_{n-1 \lra n}(Y, \mu)$ is,
    \[ \LS\parens*{ \Udw_{n - 1} } \ge \frac{1}{O\parens*{\norm*{\etch}_{\ell_\infty }}}\cdot \frac{1 - \norm*{\Jay}_\opp}{n} \cdot \parens*{1-  \theta},\]
    where we note that condition (2) in \cref{prop:is-conds} implies that the parameter $C_{n-2, 1}$ in \cref{lem:lsi} is at most $\frac{1}{O(\norm*{\etch}_\infty)}.$

    Now, \cref{cor:baked} implies that we have
    \begin{align}\label{eq:exis1}\EC(\Papx_{n-1}) &~\ge~\frac{1}{O\parens*{\norm*{\etch}_{\ell_\infty }}}\cdot \frac{1 - \norm*{\Jay}_\opp}{n} \cdot \parens*{1-  \theta} \cdot \Gap^\star(H^2),\\
    \label{eq:exis2}\EC(\Paqx_{n}) &~\ge~\frac{1}{O\parens*{\norm*{\etch}_{\ell_\infty }}}\cdot \frac{1 - \norm*{\Jay}_\opp}{n} \cdot \parens*{1-  \theta} \cdot \Gap^\star(H^2).
    \end{align}
    Now, we observe
    \[ \exp\parens*{ n + \norm*{\etch}_{\ell_1}} \ge \exp\parens*{ \frac{1}{2} \cdot \inpr*{\eks, \Jay \eks}_{\ell_2} + \inpr*{ \etch, \eks}_{\ell_2}} \ge \exp\parens*{  - \norm*{\etch}_{\ell_1}}, \]
    where we have used,
    \begin{itemize}
        \item $\max_{\eks \in \set*{\pm 1}^n} \Abs*{\inpr*{\etch, \eks}_{\ell_2}} = \norm*{\etch}_{\ell_1}$,
        \item $0 \le \inpr*{ \eks, \Jay \eks}_{\ell_2} \le \norm*{\Jay}_\opp \norm*{\eks}^2_{\ell_2} \le n$
    \end{itemize}
    Thus, we can conclude that
    \[ \min_{\eks \in \set*{+1, - 1}^n } \mu(\eks) \ge \frac{ \exp\parens*{ - \norm*{ \etch }_{\ell_1}}}{2^n \cdot \exp\parens*{ n + \norm*{\etch}_{\ell_1}}} \ge \exp\parens*{ -2n - 2\norm*{ \etch }_{\ell_1}}.\]
    Since passing from $\mu$ to $\mu \tensor \uni_{[n]}$ or $\mu_{n-1}$ shrinks the minimum measure at most by a factor of $n$, we can conclude that   \cref{eq:exis1}, \cref{eq:exis2}  together with \cref{thm:entmix} imply the theorem statement.
\end{proof}
\bibliographystyle{alpha}\bibliography{vedat}
\appendix
\section{Omitted Proofs}
\subsection{Data Processing Inequality for $\Phi$-Entropies, Proof of \cref{fac:dpi}}\label{ap:dpi}
\begin{proof}[Proof of \cref{fac:dpi}]
We write $M = \Exp_{\pi_2} \eff$. Since $\Pii$ is row-stochastic, we notice we also have $M = \Exp_{\pi_1} \Pii \eff$. Now, we can write
\[ \Ent^\Phi_{\pi_1}(\Pii \eff) = \Exp_{\omega_1 \sim \pi_1} \Phi\parens*{[\Pii \eff](\omega_1)} - M = \Exp_{\omega_1 \sim \pi_1} \Phi\parens*{ \sum_{\omega_2 \in \Omega_2} \Pii(\omega_1, \omega_2) \eff(\omega_2) } - M.\]
Using the convexity of $\Phi$ and the row-stochasticity of $\Pii$, we obtain,
\[ \Ent_{\pi_1}^\Phi( \Pii \eff) \le \Exp_{\omega_1 \sim \pi_1}\sqbr*{ \sum_{\omega_2 \in \Omega_2} \Pii(\omega_1, \omega_2) \Phi\parens*{ \eff(\omega) }} - M.\]
Now, when we expand the sum, we see that the coefficient of $\Phi\parens*{\eff(\omega_2)}$ in the expectation above is,
\[ \sum_{\omega_1 \sim \pi_1} \pi_1(\omega_1) \Pii(\omega_1, \omega_2) = \sqbr*{ \pi_1 \Pii}(\omega_2) = \pi_2(\omega_2),\]
where the last equality is due to the assumption $\pi_1 \Pii = \pi_2$. Now, using this we obtain
\[ \Ent_{\pi_1}^\Phi(\Pii \eff) \le \Exp_{\omega_2 \sim \pi_2} \Phi(\eff(\omega_2)) - M = \Ent^\Phi_{\pi_2}(\eff).\]
The second statement concerning $\CF_\Phi(\Pii\Quu)$ follows now, by observing
\begin{align*}
    \Ent_{\pi_1}(\Pii \Quu \etch)~\le~\Ent_{\pi_2}(\Quu \etch)&~\le~(1 - \CF_\Phi(\Quu)) \cdot \Ent_{\pi_3}(\etch),&&\\
    \Ent^\Phi_{\pi}(\Pii\Quu\etch)~\le~(1 - \CF_\Phi(\Pii)) \cdot \Ent_{\pi_2}(\Quu \etch) &~\le~(1 - \CF_\Phi(\Pii)) \cdot \Ent_{\pi_3}^\Phi(\etch).&&
    \end{align*}
\end{proof}
\subsection{Entropy Contraction is Controlled by the Log-Sobolev Constant, Proof of \cref{lem:lsfun}}\label{ap:lsfun}
\begin{proof}[Proof of \cref{lem:lsfun}]
    Let $\eff \in \RRp^{\Omega_2}$ we given such that $\Exp_{\mu_2} \eff^2 = 1$. We write $\etch = \eff^2 \log \eff^2$, i.e.~$\etch(\omega_2) = \eff(\omega_2)^2 \cdot \log\parens*{ \eff(\omega_2)^2}$ for all $\omega_2 \in \Omega_2$.
    \begin{claim}\label{cl:weird}
        Let $\eff \in \RRp^{\Omega_2}$ be given and $\etch$ be defined as above. Then,
        \[ \sqbr*{ \Pii \etch}(\omega_1) \ge \sqbr*{\Pii \eff^2}(\omega_1) \cdot \log\parens*{ \sqbr*{\Pii \eff^2}(\omega_1)}  + \sqbr*{\Pii \eff^2}(\omega_1) - \sqbr*{ \Pii \eff}(\omega_1)^2. \]
    \end{claim}
    Then, assuming \cref{cl:weird}, we have
    \begin{align*}
        \Ent_{\mu_1}(\Pii \eff^2) &~=~\Exp_{\omega_1 \sim \mu_1}\sqbr*{ \sqbr*{\Pii \eff^2 }(\omega_1) \cdot \log\parens*{ \sqbr*{ \Pii \eff^2}(\omega_1)}},\\
    &~\le^\varheart~\Exp_{\omega_1 \sim \mu_1}\sqbr*{ \sqbr*{\Pii \etch}(\omega_1)} - \Exp_{\omega_1 \sim \mu_1}\sqbr*{ \sqbr*{ \Pii \eff^2}(\omega_1) - \sqbr*{\Pii \eff}(\omega_1)^2},\\
        &~=~\Exp_{\omega_1 \sim \pi_1}\sqbr*{\sqbr*{ \Pii \etch}(\omega_1) } - \Exp_{\omega_1 \sim \mu_1}\sqbr*{ \sqbr*{\Pii \eff^2}(\omega_1)} - \inpr*{ \Pii \eff, \Pii \eff}_{\mu_1},\\
    &~=^\vardiamond~\Exp_{\omega_2 \sim \mu_2}[\etch(\omega_2)] - \Exp_{\omega_2 \sim \mu_2}[\eff^2(\omega_2)] + \inpr*{ \eff, \Pii^* \Pii \eff}_{\mu_2},\\
    &~=~ \Exp_{\mu_2} \etch - \inpr*{ \eff, \parens*{\Ide - \Pii^* \Pii} \eff}_{\mu_2}
    \end{align*}
    where we have used \cref{cl:weird} to obtain the inequality marked with $(\varheart)$ and that $\Pii$ is a row-stochastic operator to obtain the equality marked by $(\vardiamond)$. Now, noting that $\Exp_{\mu_2} \etch = \Ent_{\mu_2}(\eff^2)$ we can obtain,
    \[ \Ent_{\mu_1}(\Pii \eff^2) \le \Ent_{\mu_2}(\eff^2) - \inpr*{ \eff, \parens*{\Ide - \Pii^* \Pii} \eff}_{\mu_2}.\]
    Now appealing to the definition of the log-Sobolev constant $\LS(\Pii^* \Pii)$ (\cref{eq:lsc-def}) we obtain
    \[ \Ent_{\mu_1}(\Pii \eff^2) \le (1 - \LS(\Pii^* \Pii) ) \cdot \Ent_{\mu_2}(\eff^2).\]
    Since for any $\gee \in \RRp^{\Omega}$ exists $\eff \in \RRp^{\Omega}$ such that $\gee = \eff^2$, we have shown $\EC(\Pii) \ge \LS(\Pii^* \Pii)$, as $\EC(\Pii)$ is the largest constant $C$ such that the inequality
     \[ \Ent_{\mu_1}(\Pii \gee) \le (1 - C ) \cdot \Ent_{\mu_2}(\gee),\]
     holds for all $\gee \in \RRp^{\Omega}$. Now, we prove the claim \cref{cl:weird},
     \begin{proof}[Proof of \cref{cl:weird}]
     We will make use of the following inequality, \cite[Lemma 5]{Miclo97}
     \begin{equation}\label{eq:miclo}
         (t + s) \log (t + s) \ge  t \log t + s (1 + \log t) + \parens*{ \sqrt{t + s} - \sqrt t}^2~~\textrm{for all}~~t \ge 0~\textrm{and}~s \ge - t
     \end{equation}
     Now, writing $t = \sqbr*{ \Pii \eff^2}(\omega_1)$, we have
     \begin{align*}
         \sqbr*{\Pii \etch}(\omega_1) &~=~\sum_{\omega_2 \in \Omega_2} \Pii(\omega_1, \omega_2) \cdot (\eff^2(\omega_2) \log \eff^2(\omega_2)),\\
         &~=~\sum_{\omega_2 \in \Omega_2} \Pii(\omega_1, \omega_2) \cdot (\eff^2(\omega_2 + t - t) \log \eff^2(\omega_2+ t - t)),\\
         &~\ge^{\star}~\sum_{\omega_2 \in \Omega_2}\Pii(\omega_1, \omega_2) \cdot \parens*{ t \log t + (\eff^2(\omega_2)  -t)(1 + \log t) + \parens*{\eff(\omega_2) - \sqrt{t}}^2},\\
         &~=^\varheart~\sqbr*{\Pii \eff^2}(\omega_1) \cdot \log\parens*{ \sqbr*{ \Pii \eff^2}(\omega_1)} + \sum_{\omega_2 \in \Omega_2} \Pii(\omega_1, \omega_2) \cdot \parens*{ \eff(\omega_2) - \sqrt t}^2
     \end{align*}
     where we have used \cref{eq:miclo} to obtain the inequality marked with $(\star)$ and that $\sum_{\omega_2} \Pii(\omega_1, \omega_2) \eff^2(\omega_2) = t$ to obtain the equality marked with $(\star)$ -- as this implies that the second term in the RHS of the inequality above vanishes when one expands the sum. Now expanding the second term further, we have obtain
     \begin{equation}\label{eq:last-weird}
         \sqbr*{ \Pii \etch}(\omega_1) \ge \sqbr*{\Pii\eff^2}(\omega_1) \cdot \log\parens*{ \sqbr*{\Pii \eff^2}(\omega_1)} + \underbrace{2 \parens*{\sqbr*{\Pii \eff^2}(\omega_1)} - 2\sqrt{ \sqbr*{\Pii \eff^2}(\omega_1) } \cdot \sqbr*{ \Pii \eff}(\omega_1)}_{\tau}.
     \end{equation}
     Notice that, we have
     \[ \tau- \parens*{\sqbr*{\Pii \eff^2}(\omega_1) - \sqbr*{\Pii \eff}(\omega_1)^2} = \parens*{ \sqrt{\sqbr*{ \Pii \eff^2}(\omega_1)} - \sqrt{\sqbr*{\Pii \eff}(\omega_1)}}^2 \ge 0\]
     which in conjunction with \cref{eq:last-weird} implies the desired inequality.
     \end{proof}
 \end{proof}
 \subsection{Variance Contraction is Controlled By Links, Proof of \cref{eq:varcont}}\label{app:varcont}
 \begin{proof}[Proof of \cref{eq:varcont}]
    Recall that $\Ent^{\Phi}_{\bullet}(\bullet)$ is just the \ref{eq:v-def} functional $\Var_{\bullet}(\bullet)$ for $\Phi(t) = t^2$. Thus, $\lc_\Phi(\widehat\omega)$ is the smallest constant $C$ for which the inequality,
    \begin{equation}\label{eq:gyoza}
\Var_{\pi_1^{(\widehat\omega)}}\parens*{ \Upo_{\widehat\omega, 1 \to n'} \gee} \le C \cdot \Var_{\pi_{n'}^{(\widehat\omega)}}(\gee).      
    \end{equation} 
    where $n' = n - |\widehat\omega|$. In particular, writing $c = \Exp_{\pi_{n'}^{(\widehat\omega)}} \gee$. Then, we have
    \begin{align*}
        \Var_{\pi_{n'}^{(\widehat\omega)}}(\gee) &~=~\inpr*{\gee - c \cdot \one, \gee - c \cdot \one}_{\pi_{n'}^{(\widehat\omega)}},\\
        \Var_{\pi_1^{(\widehat\omega)}} &~=~\inpr*{\Upo_{\widehat\omega, 1 \to n'}(\gee - c \cdot \one), \Upo_{\widehat\omega}(\gee - c \cdot \one)}_{\pi_1^{(\widehat\omega)}}.
    \end{align*}
    In particular, by replacing $\gee$ with $\etch := \gee - c \cdot \one$, we can observe that \cref{eq:gyoza} is equivalent to
    \[ \inpr*{ \etch, \Duw_{\widehat\omega, n' \lra 1} \etch}_{\pi_{n'}^{
(\widehat\omega)}} = \inpr*{\Upo_{\widehat\omega, 1 \to n'}\etch, \Upo_{\widehat\omega}\etch }_{\pi_1^{(\widehat\omega)}} \le C \cdot \norm*{ \etch }^2_{\pi_{n'}^{(\widehat\omega)}}.\]
In particular, by \cref{fac:cf-baby} the best $C$ that satisfies the inequality is simply $\lambda_2(\Duw_{\widehat\omega, n' \lra 1})$. Notice that, by \cref{fac:switcheroo} 
\[ \lambda_2(\Duw_{\widehat\omega, n' \lra 1}) = \lambda_2\parens*{\Doo_{\widehat\omega, n' \to 1}\Upo_{\widehat\omega, 1 \to n'}} = \lambda_2\parens*{ \Upo_{\widehat\omega, 1 \to n'}\Upo_{\widehat\omega, n' \to 1}} = \lambda_2\parens*{\Udw_{\widehat\omega, 1 \lra n'}}.\]
Now, a direct computation shows that for all $x, y \in X_{\widehat\omega}^{(1)}$,
\[ \Udw_{\widehat\omega, 1 \lra n'}(x, y) = \frac{ \one[x = y]}{n'} + \frac{ \one[x \ne y] \cdot \Pr_{\omega \sim \pi}\sqbr*{ \omega \supset \widehat\omega \sqcup \set{x, y} \mid \omega \supset \widehat\omega }}{n'}.\]
In particular, recalling the definition of the \ref{eq:link-def} $\Emm_{\widehat\omega}$,
\[ \Udw_{\widehat\omega, 1 \lra n'} = \frac{\Ide}{n'} + \frac{n'- 1}{n'} \cdot \Emm_{\widehat\omega},\]
i.e.~$\lambda_2(\Udw_{\widehat\omega, 1 \lra n'}) = \frac{1}{n'} + \frac{n' - 1}{n'} \cdot \lambda_2(\Emm_{\widehat\omega})$. Thus the statement follows.
\end{proof}
\subsection{Stationarity and Reversibility of Expanderized Up-Down Walk for Non-Squared $H$, Proof of \cref{prop:stat}}\label{ap:stat}
\begin{proof}[Proof of \cref{prop:stat}]
  Let $\widehat\omega, \widetilde\omega \in X^{(\ell)}$ be given, with $\typ(\widehat\omega) = S$ and $\typ(\widetilde\omega) = T$.
  Notice that $\Papx_{\ell \lra n}$ makes a transition from $\widehat\omega$ to $\widetilde\omega$ under two conditions (i) $S \sim_H T$ and (ii) we sample a face $\omega \in X^{(n)}$ in the first step of the algorithm such that $\omega_T = \widetilde\omega$. Thus, we have
  \begin{equation}\label{eq:eqstuf} \Papx_{\ell \lra n}(\widehat\omega, \widetilde\omega) = \frac{\one[S \sim_H T]}{k} \cdot \Pr_{\omega \sim \pi_n}\sqbr*{ \omega_T = \widetilde\omega \mid \omega_S = \widehat\omega}. \end{equation}
  With this, we compute
  \begin{align*} \sqbr*{\pi_\ell \Papx_{\ell \lra n}}(\widetilde\omega) &~=~\sum_{\widehat\omega \in X^{(\ell)}} \pi_{\ell}(\widehat\omega) \cdot  \frac{\one\sqbr*{ \typ(\widehat\omega) \sim_H T}}{k} \cdot \Pr_{\omega \sim \pi_n}\sqbr*{\omega_T = \widetilde\omega \mid \omega_{\typ(\widehat\omega)} = \widehat\omega},\\
  &~=~\frac{1}{\binom{n}{\ell}}\sum_{S \in \binom{n}{\ell}} \frac{\one[S \sim_H T]}{k} \cdot \sum_{\widehat\omega \in X[S]} \Pr_{\omega \sim \pi_n}[\omega_S = \widehat\omega] \cdot \Pr_{\omega \sim \pi_n}\sqbr*{\omega_T = \widehat\omega \mid \omega_S = \widehat\omega },\\
  &~=~~\frac{1}{\binom{n}{\ell}}\sum_{S \in \binom{n}{\ell}} \frac{\one[S \sim_H T]}{k} \cdot \sum_{\widehat\omega \in X[S]} \Pr_{\omega \sim \pi_n}\sqbr*{ \omega_S = \widehat\omega~~\textrm{and}~~\omega_T= \widetilde\omega},\\
  &~=~\frac{\Pr_{\omega \sim \pi_n}[\omega_T = \widetilde\omega]}{\binom{n}{\ell}} \cdot \sum_{ S \in \binom{n}{\ell}} \frac{\one[S \sim_H T]}{k} \sum_{\widehat\omega \in X[S]} \Pr[\omega_S = \widehat\omega \mid \omega_T = \widetilde\omega ],\\
  &~=~\frac{\Pr_{\omega \sim \pi_n}[\omega_T = \widetilde\omega]}{\binom{n}{\ell}},\\ &~=~ \pi_\ell(\widetilde\omega),
  \end{align*}
  where we have used multiple times that $\pi_\ell(\widetilde\omega) =  \binom{n}{\ell}^{-1} \cdot \Pr[\omega_T = \widetilde\omega]$ for any $\widetilde\omega \in X^{(\ell)}$ in a partite complex with $\typ(\widetilde\omega) = T$. The last inequality follows since the inner some over $\widehat\omega$ sums to 1 and so does the outer sum. 

  Using \cref{eq:eqstuf} we can also verify the following detailed balance conditions, since
  \[ \pi_\ell(\widehat\omega) \Papx_{\ell \lra n}(\widehat\omega, \widetilde\omega) = \frac{\one[S \sim_H T]}{k} \cdot \frac{\Pr_{\widehat\omega \sim \pi_n}[\omega_T = \widetilde\omega~~\textrm{and}~~\omega_S = \widehat\omega]}{\binom{n}{\ell}} = \pi_{\ell}(\widetilde\omega) \Papx_{\ell \lra n}(\widetilde\omega, \widehat\omega), \]
  where we have assumed $\typ(\widetilde\omega) = T$ and $\typ(\widehat\omega) = S$.
\end{proof}

\end{document}